\documentclass[journal,twoside,web]{ieeecolor}

\usepackage{generic}
\usepackage{cite}
\usepackage{amsmath,amssymb,amsfonts}
\usepackage{dsfont}
\usepackage{graphicx}
\usepackage{enumerate}
\let\labelindent\relax
\usepackage{enumitem}
\usepackage{graphics}
\usepackage[font={small}]{caption}
\usepackage{multirow}
\usepackage{dsfont}
\usepackage{epsfig}
\usepackage{cite}
\usepackage{apptools}

\newcommand{\argmax}[1]{\underset{#1}
{\operatorname{arg}\,\operatorname{max}}\;}

\newcommand{\bs}{\boldsymbol}

\def\mbb{\mathbb}
\def\mcal{\mathcal}
\def\mc{\mathcal}

\def\tbfp{\textbf{p}}

\newtheorem{theorem}{Theorem}[section]

\newtheorem{lemma}{Lemma}[section]
\newtheorem{proposition}{Proposition}[section]
\newtheorem{remark}{Remark}[section]

\newtheorem{definition}{Definition}[section]

\begin{document}

\title{Reinforcement Strategies in General Lotto Games}

\author{Keith Paarporn, Rahul Chandan, Mahnoosh Alizadeh, Jason R. Marden\thanks{K. Paarporn is with the Department of Computer Science at the University of Colorado, Colorado Springs. R. Chandan is with Amazon Robotics. M. Alizadeh and J. R. Marden are with the Department of Electrical and Computer Engineering at the University of California, Santa Barbara, CA. Contact: \texttt{kpaarpor@uccs.edu}; \texttt{rcd@amazon.com}; \texttt{ \{alizadeh,jrmarden\}@ucsb.edu} This work is supported by UCOP Grant LFR-18-548175, ONR grant \#N00014-20-1-2359, and AFOSR grants \#FA9550-20-1-0054 and \#FA9550-21-1-0203. This manuscript extends a previous conference version \cite{chandan2022strategic}. }
}

\maketitle

\begin{abstract}
    Strategic decisions are often made over multiple periods of time, wherein decisions made earlier impact a competitor's  success in later stages. In this paper, we study these dynamics in General Lotto games, a class of models describing the competitive allocation of resources between two opposing players. We propose a two-stage formulation where one of the players has reserved resources that can be strategically  \emph{pre-allocated} across the battlefields in the first stage of the game as reinforcements.  The players then simultaneously allocate their remaining \emph{real-time} resources, which can be randomized, in a decisive final stage. Our main contributions provide complete characterizations of the optimal reinforcement strategies and resulting equilibrium payoffs in these multi-stage General Lotto games.  Interestingly, we determine that real-time resources are at least twice as effective as reinforcement resources when considering equilibrium payoffs. 
\end{abstract}

\section{Introduction}\label{sec:intro}


System planners must make investment decisions to mitigate the risks posed by disturbances or adversarial interference.  In many practical settings, these investments are made and built over time, leading up to a decisive point of conflict.  Security measures in cyber-physical systems and public safety are deployed and accumulated over long periods of time. Attackers can consequently use knowledge of the pre-deployed elements to identify vulnerabilities and exploits in the defender's strategy \cite{etesami2019dynamic,brown2006defending,zhang2013protecting}. Many types of contests involves deciding how much effort to exert over multiple rounds of competition \cite{konrad2009strategy,Aidt_2019,sela2023resource,yildirim2005contests,kovenock2008electoral,dechenaux2015survey}.

Indeed, investment decisions are dynamic, where early investments affect how successful a competitor is at later points in time. Many of these scenarios involve the strategic allocation of resources, exhibiting trade-offs between the costs of investing resources in earlier periods and reserving resources for later stages. In particular, an adversary is often able to learn how the resources were allocated in the earlier periods and can exploit this knowledge in later periods.


In this manuscript, we seek to characterize the interplay between early and late resource investments. We study these elements in General Lotto games, a game-theoretic framework that describes the competitive allocation of resources between opponents. The General Lotto game is a popular variant of the classic Colonel Blotto game, wherein two budget-constrained players, $A$ and $B$, compete over a set of valuable battlefields. The player that deploys more resources to a battlefield wins its associated value, and the objective for each player is to win as much value as possible. Outcomes in the standard formulations are determined by a single simultaneous allocation of resources, i.e. they are typically studied as one-shot games.

The formulations considered in this paper focus on a multi-stage version of the General Lotto game where one of the players can reinforce various battlefields before the competition begins by pre-allocating resources to battlefields; hence, we refer to these reinforcement strategies as pre-allocation strategies.   More formally,  our analysis is centered on the following multi-stage scenario: Player $A$ is endowed with $P \geq 0$ resources to be pre-allocated, and both players possess real-time resources $R_A, R_B \geq 0$ to be allocated at the time of competition. In the first stage, player $A$ decides how to  deploy the pre-allocated resources $P$ over the battlefields. The pre-allocation decision is binding and known to player $B$. In the final stage, both players engage in a General Lotto game where they simultaneously decide how to deploy their real-time resources, and payoffs are subsequently derived. 

The pre-allocated resources may represent, for example, the installation of anti-virus tools on system servers. The capabilities of anti-virus software are typically static and well-known, and thus a potential attacker would have knowledge about the system's base level of defensive capability. However, the attacker would not generally have knowledge about the system's placement of intrusion-detection systems, which are often dynamic and part of a ``moving target defense" strategy \cite{chowdhary2019adaptive,zhu2013game,zhuang2014towards}. Moreover, attackers' strategies must be unpredictable in an attempt to exploit defenses. Thus, the use of real-time resources in our model represents such dynamic and unpredictable strategies. A full summary of our contributions is provided below.

\smallskip \noindent \textbf{Our Contributions:} Our main contribution in this paper is a full characterization of equilibrium strategies and payoffs to both players in the aforementioned two-stage General Lotto game (Theorem \ref{thm:equilibrium_characterization}). By characterizing these optimal reinforcement strategies, we are able to provide Pareto frontiers for player $A$ as one balances a combination of real-time and pre-allocated resources (Lemma \ref{lem:level_set}).  Interestingly, Theorem \ref{thm:ratio} demonstrates that real-time resources are at least twice as effective as pre-allocated resources when considering the equilibrium payoff of player $A$.


Our second set of results in this manuscript focus on the optimal investment levels of pre-allocated and real-time resources. Rather than player $A$ being equipped with a fixed budget of resources $(P,R_A)$, we rather consider a setting where player $A$ has a monetary budget $M_A$ and each type of resource is associated with a given per-unit cost. Building upon the above characterization of the optimal reinforcement strategies in Theorem~\ref{thm:equilibrium_characterization}, in Theorem~\ref{thm:investment} we characterize the optimal investment strategies for this per-unit cost variant of the two-stage General Lotto game.  This  provides an understanding of the precise combination of pre-allocated and real-time resources that optimize player $A$'s equilibrium payoff.


Our last contribution focuses on a variant of this General Lotto game where both players can employ pre-allocated resources.  In particular, we consider a scenario where player $B$ is able to respond to player $A$'s pre-allocation with its own pre-allocated resources, before engaging in the final-stage General Lotto game. This is formulated as a Stackelberg game, where both players have  monetary budgets $M_A,M_B$ and per-unit costs for investing in the two types of resources. We fully characterize the Stackelberg equilibrium (Proposition \ref{prop:stack_equil}), which highlights that having the opportunity to respond to an opponent's early investments can significantly improve one's eventual performance.


\smallskip \noindent \textbf{Related works:} This manuscript takes steps towards understanding the competitive  allocation of resources in multi-stage scenarios. There is widespread interest in this research objective that involves the analysis of zero-sum games \cite{Nayyar_2013,Kartik_2021asymmetric,Li_2020}, differential or repeated games \cite{Isaacs_1965,VonMoll_2020}, and Colonel Blotto games \cite{Vu_2019combinatorial,Aidt_2019,Leon_2021,shishika2022dynamic}. The goal of many of these works is to develop tools to compute decision-making policies for agents in adversarial and uncertain environments. In comparison, our work provides explicit, analytical characterizations of equilibrium strategies, which draws sharper insights that relate the players' performance with the various elements of adversarial interaction. As such, our work is related to a recent research thread in which allocation decisions are made over multiple stages \cite{Kovenock_2012,Gupta_2014a,Gupta_2014b,Leon_2021,Vu_2019combinatorial,Paarporn2021strategically,shishika2022dynamic,chandan2023art}. 

Our work also draws significantly from the primary literature on Colonel Blotto and General Lotto games \cite{Gross_1950,Roberson_2006,Kovenock_2020,Vu_EC2021}. In particular, the simultaneous-move subgame played in the final stage of our formulations was first proposed by Vu and Loiseau \cite{Vu_EC2021}, and is known as the \emph{General Lotto game with favoritism} (GL-F). Favoritism refers to the fact that pre-allocated resources provide an incumbency advantage to one player's competitive chances. Their work establishes existence of equilibria and develops computational methods to calculate them to arbitrary precision. However, this prior work considers pre-allocated resources as exogenous parameters of the game. In contrast, we model the deployment of pre-allocated resources as a strategic element of the competitive interaction.  Furthermore, we provide the first analytical characterizations of equilibria and the corresponding payoffs in GL-F games. 

\section{Problem formulation}\label{sec:model}

\begin{figure*}[t]
    \centering
    \includegraphics[scale=0.25]{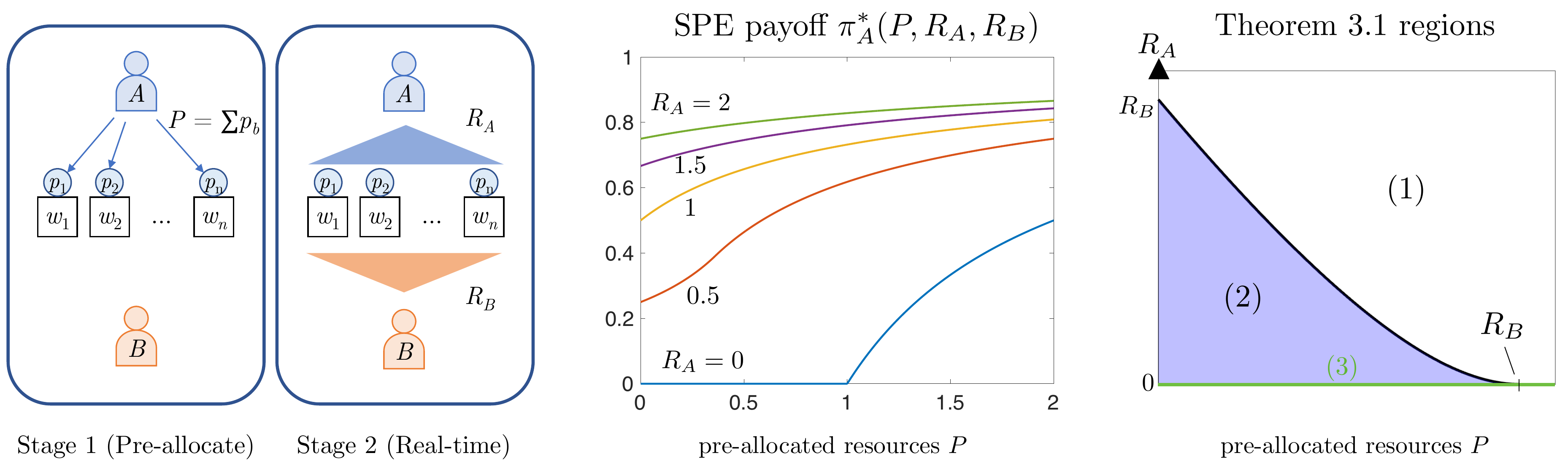}
    \caption{(Left) The two-stage General Lotto game with Pre-allocations (GL-P). Players $A$ and $B$ compete over $n$ battlefields, whose valuations are given by $\{w_b\}_{b=1}^n$. In Stage 1, player $A$ decides how to deploy $P$ pre-allocated resources to the battlefields. Player $B$ observes the deployment. In Stage 2, the players simultaneously decide how to deploy their real-time resources $R_A$ and $R_B$ and final payoffs are determined. (Center) This plot shows the SPE payoff to player $A$ under varying resource endowments (Theorem \ref{thm:equilibrium_characterization}). Obtaining more pre-allocated resources improves the payoff with decreasing marginal returns. Here, we have fixed $R_B = 1$. (Right) The characterization of the SPE payoff is broken down into three separate cases in the game's parameters. These are shown as the three regions in this plot, here parameterized by $P$ and $R_A$, which correspond to the items in Theorem \ref{thm:equilibrium_characterization}. }
    \label{fig:GLP_figure}
\end{figure*}

The \emph{General Lotto game with pre-allocations} (GL-P) is a two-stage game with players $A$ and $B$, who compete over a set of $n$ battlefields, denoted as $\mcal{B} = \{1,\ldots,n\}$. Each battlefield $b\in\mcal{B}$ is associated with a known valuation $w_b > 0$, which is common to both players. Player $A$ is endowed with a pre-allocated resource budget $P > 0$ and a real-time resource budget $R_A>0$. Player $B$ is endowed with a  budget $R_B > 0$ of real-time resource, but no pre-allocated resources. The two stages are played as follows:

\smallskip \noindent -- \emph{Stage 1 (pre-allocation):} Player $A$ decides how to allocate her $P$ pre-allocated resources to the battlefields, i.e., it selects a vector $\bs{p} = (p_1,\ldots,p_n) \in \Delta_n(P) := \{\bs{p}' \in \mbb{R}_+^n : \|\bs{p}'\|_1 = P \}$. We term the vector $\bs{p}$ as player $A$'s \emph{pre-allocation profile}. No payoffs are derived in Stage 1, and $A$'s choice $\bs{p}$ becomes binding and common knowledge.

\smallskip \noindent -- \emph{Stage 2 (decisive point of conflict):} Players $A$ and $B$ then compete in a simultaneous-move sub-game with their real-time resource budgets $R_A$, $R_B$. Here, both players can randomly allocate these resources as long as their expenditure does not exceed their budgets in expectation. Specifically, a strategy for player $i \in \{A,B\}$ is an $n$-variate (cumulative) distribution $F_i$ over allocations $\bs{x}_i \in \mbb{R}_+^n$ that satisfies
\begin{equation}\label{eq:lotto_budget_constraint}
    \mbb{E}_{\bs{x}_i \sim F_i} \left[ \sum_{b\in\mcal{B}} x_{i,b} \right] \leq R_i.
\end{equation}

\noindent We use $\mcal{L}(R_i)$ to denote the set of all strategies $F_i$ that satisfy \eqref{eq:lotto_budget_constraint}.  Given that player $A$ chose $\bs{p}$ in Stage 1, the expected payoff to player $A$ is given by
\begin{equation}
    U_A(\bs{p},F_A,F_B) := \mbb{E}_{\substack{\bs{x}_A \sim F_A \\ \bs{x}_B \sim F_B}}\left[ \sum_{b\in\mcal{B}} w_b \! \cdot \! \mathds{1}\{ x_{A,b} + p_b \geq qx_{B,b} \} \right]
\end{equation}
where $\mathds{1}\{\cdot\}$ is the usual indicator function taking a value of 1 or 0.\footnote{The tie-breaking rule (i.e., deciding who wins if $x_{A,b}+p_b=x_{B,b}$) can be assumed to be arbitrary, without affecting any of our results. This property is common in the General Lotto literature, see, e.g., \cite{Kovenock_2020,Vu_EC2021}.}  In words, player $B$ must overcome player $A$'s pre-allocated resources $p_b$ as well as player $A$'s allocation of real-time resources $x_{A,b}$ in order to win battlefield $b$.  The parameter $q > 0$ is the relative quality of player $B$'s real-time resources against player $A$'s resources. For simpler exposition, we will simply set $q = 1$, noting that all of our results are easily attained for any other value of $q$. The payoff to player $B$ is $U_B(\bs{p},F_A,F_B) = 1 - U_A(\bs{p},F_A,F_B)$, where we assume without loss of generality that  $\sum_{b\in\mcal{B}} w_b = 1$.

\smallskip Stages 1 and 2 of GL-P are illustrated in Figure \ref{fig:GLP_figure}.  We specify an instantiation of the game as $\text{GL-P}(P,R_A,R_B,\bs{w})$. We focus on the subgame-perfect equilibrium solution concept.


\begin{definition}\label{def:SPE}
    A profile $(\bs{p}^*,F_A^*(\bs{p}),F_B^*(\bs{p}))$ where $\bs{p}^* \in \Delta_n(P)$ and $F_i^*(\bs{p}) : \Delta_n(P) \rightarrow \mcal{L}(R_i)$, for $i=A,B$, is a \emph{subgame-perfect equilibrium (SPE)} if the following conditions hold.
    \begin{itemize}
        \item For any $\bs{p} \in \Delta_n(P)$, $(F_A^*(\bs{p}),F_B^*(\bs{p}))$ constitutes a Nash equilibrium of the Stage 2 subgame:
        \begin{equation}
            \begin{aligned}
                U_A(\bs{p},F_A^*(\bs{p}),F_B^*(\bs{p})) &\geq U_A(\bs{p},F_A,F_B^*(\bs{p})) \\
                \text{and } U_B(\bs{p},F_A^*(\bs{p}),F_B^*(\bs{p})) &\geq U_B(\bs{p},F_A^*(\bs{p}),F_B)
            \end{aligned}
        \end{equation}
        for any $F_A \in \mcal{L}(R_A)$ and $F_B \in \mcal{L}(R_B)$.
        \item The pre-allocation $\bs{p}^*$ satisfies
        \begin{equation}
            U_A(\bs{p}^*,F_A^*(\bs{p}^*),F_B^*(\bs{p}^*)) \geq U_A(\bs{p},F_A^*(\bs{p}),F_B^*(\bs{p}))
        \end{equation}
        for any $\bs{p} \in \Delta_n(P)$.
    \end{itemize}

\end{definition}


In an SPE, the players select their Stage 2 strategies conditioned on player $A$'s choice of pre-allocation $\bs{p}$ in Stage 1, such that $F_A^*(\bs{p}), F_B^*(\bs{p})$ forms a Nash equilibrium of the one-shot subgame of Stage 2. We stress the importance of the common knowledge assumption for the pre-allocation choice $\bs{p}$ before Stage 2 -- over time, an opponent is likely to learn the placement of past resources and would be able to exploit this knowledge at a later point in time. The second condition in the above definition asserts that player $A$'s SPE pre-allocation $\bs{p}^*$ in Stage 1 optimizes its equilibrium payoff in the subsequent Stage 2 subgame. 

We remark that the Stage 2 subgame has been studied in the recent literature, where it is termed a \emph{General Lotto game with Favoritism} \cite{Vu_EC2021}. We denote it as $\text{GL-F}(\bs{p},R_A,R_B)$. There, a pre-allocation vector $\bs{p}$ is viewed as an exogenous fixed parameter, whereas in our GL-P formulation, it is an endogenous strategic choice. It is established in \cite{Vu_EC2021} that a Nash equilibrium exists and its payoffs are unique for any instance of $\text{GL-F}(\bs{p},R_A,R_B)$. Consequently, the players' SPE payoffs in our GL-P game are necessarily unique. We will denote $\pi^*_i(P,R_A,R_B) := U_i(\bs{p}^*,F_A^*(\bs{p}^*),F_B^*(\bs{p}^*))$, $i\in\{A,B\}$, as the players' payoffs in an SPE when the dependence on the vector $\bs{w}$ is clear.

While \cite{Vu_EC2021} provided numerical techniques to compute an equilibrium of $\text{GL-F}(\bs{p},R_A,R_B)$ to arbitrary precision, analytical characterizations of them (e.g. closed-form expressions) were not provided. In the next section, we develop techniques to derive such characterizations, as they are required to precisely express the SPE  of the GL-P game.

%


\section{Equilibrium characterizations}\label{sec:results}

In this section, we present our main results regarding the characterization of players' SPE payoffs in the GL-P game. These results highlight the relative effectiveness of pre-allocated vs real-time resources.


\subsection{Main results}
The result below 
provides an explicit characterization of the players' payoffs in an SPE of the two-stage GL-P game.


\begin{theorem}\label{thm:equilibrium_characterization}
    Consider the game $\text{GL-P}(P,R_A,R_B,\bs{w})$.  Player $A$'s payoff $\pi^*_A(P,R_A,R_B)$ in a SPE is given as follows:
    
    \begin{enumerate}
        \item If $R_B \leq P$, or $R_B > P$ and $R_A \geq \frac{2(R_B-P)^2}{P+2(R_B-P)}$, then $\pi^*_A(P,R_A,R_B)$ is
        \begin{equation} \label{eq:allcase1_payoff}
            1 - \frac{R_B}{2R_A}\left(\frac{R_A + \sqrt{R_A(R_A+2P)}}{P+R_A+\sqrt{R_A(R_A+2P)}} \right)^2 .
        \end{equation}
        \item If $R_B > P$ and $0 < R_A < \frac{2(R_B-P)^2}{P+2(R_B-P)}$, then $\pi^*_A(P,R_A,R_B)$ is
        \begin{equation} \label{eq:allcase2_payoff}
            \frac{R_A}{2(qR_B-P)}.
        \end{equation}
        \item If $R_A = 0$, then $\pi^*_A(P,R_A,R_B)$ is
        \begin{equation}\label{eq:allcase3_payoff}
            \left(1 - \min\left\{\frac{R_B}{P},1\right\} \right)
        \end{equation}
    \end{enumerate}
    Player $B$'s SPE payoff is given by $\pi^*_B(P,R_A,R_B) = 1 - \pi^*_A(P,R_A,R_B)$. In all instances, player $A$'s SPE pre-allocation is $\bs{p}^* = \bs{w}\cdot P$.
\end{theorem}

A visualization of the parameter regimes of the three cases above is shown in the right Figure \ref{fig:GLP_figure}. Note that the standard General Lotto game (without pre-allocations, \cite{Hart_2008}) is included as the vertical axis at $P=0$. An illustration of the SPE payoffs to player $A$ is shown in the center plot of Figure \ref{fig:GLP_figure}. We notice that given sufficiently high amount of pre-allocated resources (i.e. $P > R_B$), player $A$ can attain a positive payoff even without any real-time resources ($R_A=0$). For $P < R_B$, player $A$ receives zero payoff since player $B$ can simply exceed the pre-allocation on every battlefield. Observe that the SPE payoff $\pi_A^*(P,R_A,R_B)$ exhibits diminishing marginal returns in $R_A$ and in $P$ for larger values of $P$, but is not in general a concave function in $P$ -- see the $R_A = 0.5$ curve, which has an inflection.

\noindent\textbf{Proof approach and outline:} The derivation of the SPE payoffs in Theorem \ref{thm:equilibrium_characterization} follows a backwards induction approach. First, for any fixed pre-allocation vector $\bs{p}$, one characterizes the equilibrium payoff of the Stage 2 sub-game $\text{GL-F}(\bs{p},R_A,R_B)$. We denote this payoff as
\begin{equation}
    \pi_A(\bs{p},R_A,R_B) := U_A(\bs{p},F_A^*(\bs{p}),F_B^*(\bs{p}))
\end{equation}
where $F_A^*(\cdot),F_B^*(\cdot)$ satisfies the first condition of Definition \ref{def:SPE}. Then, the SPE payoff is calculated by solving the following optimization problem,
\begin{equation}\label{eq:p_optimization}
    \pi_A^*(P,R_A,R_B) = \max_{\bs{p}\in\Delta_n(P)} \pi_A(\bs{p},R_A,R_B).
\end{equation}


The following proof outline is taken to derive the SPE strategies and payoffs.

\vspace{1mm}

\noindent\textbf{\underline{Part 1}:} We first detail analytical methods to derive the equilibrium payoff $\pi_A(\bs{p},R_A,R_B)$ to the second stage subgame $\text{GL-F}(\bs{p},R_A,R_B)$. 

\vspace{1mm}

\noindent\textbf{\underline{Part 2}:} We show that  $\bs{p}^* = \bs{w}\cdot P$ is an SPE pre-allocation, i.e. it solves the optimization problem \eqref{eq:p_optimization}.

\vspace{1mm}

\noindent\textbf{\underline{Part 3}:} We derive the analytical expressions for $\pi_A^*$ reported in Theorem \ref{thm:equilibrium_characterization}.

Each one of the three parts has a corresponding Lemma that we present in the following subsection.

\subsection{Proof of Theorem \ref{thm:equilibrium_characterization}}

\noindent\textbf{\underline{Part 1:}} The recent work of Vu and Loiseau \cite{Vu_EC2021} provides a method to derive an equilibrium of the  General Lotto game with Favoritism $\text{GL-F}(\bs{p},R_A,R_B)$.  This method involves solving the following system\footnote{The problem settings considered in \cite{Vu_EC2021} are more general, which considers two-sided favoritism (i.e. $p_b < 0$ for some $b$). However, exact closed-form solutions were not provided. The paper \cite{Vu_EC2021} provided computational approaches to calculate an equilibrium to arbitrary precision.} of two equations for two unknown variables $(\kappa_A,\kappa_B) \in \mbb{R}_{++}^2$: 
\begin{equation}\label{eq:SOE}
    \begin{aligned}
        R_A \!=\! \sum_{b=1}^n \frac{[h_b(\kappa_A,\kappa_B) - p_b]^2}{2w_b\kappa_B}, \ 
        R_B \!= \sum_{b=1}^n \frac{h_b^2(\kappa_A,\kappa_B) - p_b^2}{2w_b\kappa_A}
    \end{aligned}
\end{equation}
where $h_b(\kappa_A,\kappa_B) := \min\{  w_b\kappa_B, w_b\kappa_A + p_b\}$ for $b \in \mcal{B}$. The two equations above correspond to the expected budget constraint \eqref{eq:lotto_budget_constraint} for both players. There always exists a solution $(\kappa_A^*,\kappa_B^*) \in \mbb{R}_{++}^2$ to this system \cite{Vu_EC2021}, which allows one to calculate the following equilibrium payoffs. 

\begin{lemma}[Adapted from \cite{Vu_EC2021}]\label{lem:SOE}
    Suppose $(\kappa_A^*,\kappa_B^*) \in \mbb{R}_{++}^2$ solves \eqref{eq:SOE}.  Let $\mcal{B}_1 := \{b\in\mcal{B} : h_b(\kappa_A^*,\kappa_B^*) = w_b\kappa_B\}$ and $\mcal{B}_2 = \mcal{B}\backslash \mcal{B}_1$.
    Then there is a Nash equilibrium $(F_A^*,F_B^*)$ of $\text{GL-F}(\bs{p},R_A,R_B)$ where player $A$'s equilibrium payoff is given by 
    \begin{equation}\label{eq:playerA_payoff}
        \begin{aligned}
            \pi_A(\bs{p},R_A,R_B) &= \sum_{b\in\mcal{B}_1} \!w_b\!\left[1 - \frac{\kappa_B^*}{2\kappa_A^*}\left(1 - \frac{p_i^2}{(w_b\kappa_B)^2} \right) \right] \\
            &\quad + \sum_{b\in\mcal{B}_2} w_b \frac{\kappa_A^*}{2\kappa_B^*}
        \end{aligned}
    \end{equation}
    and the equilibrium payoff to player $B$ is $\pi_B(\bs{p},R_A,R_B) = 1 - \pi_A(\bs{p},R_A,R_B)$.
    
\end{lemma}

Lemma \ref{lem:SOE} provides an expression for $\pi_A(\bs{p},R_A,R_B)$ in terms of a solution $(\kappa_A^*,\kappa_B^*)$ to the system of equations \eqref{eq:SOE}. However, in order to study the optimization \eqref{eq:p_optimization}, we need to be able to either find closed-form expressions for the solution  $(\kappa_A^*,\kappa_B^*)$ in terms of the defining game parameters $\bs{p},R_A,R_B,\bs{w}$, or establish certain properties about the payoff function \eqref{eq:playerA_payoff}, such as concavity in $\bs{p}$. Unfortunately, we find that this function is not generally concave for $\bs{p} \in \Delta_n(P)$. Our approach in Part 2 is to show that it is always increasing in the direction pointing to $\bs{p}^*$.

\noindent\textbf{\underline{Part 2:}} This part of the proof is devoted to showing that $\bs{p}^* = \bs{w}\cdot P$ is an SPE pre-allocation for player $A$. This divides the total pre-allocated resources $P$ among the battlefields proportionally to their values $w_b$, $b \in \mcal{B}$.

\begin{lemma}\label{lem:global_maximizer}
    The vector $\bs{p}^* = \bs{w}\cdot P$ is an SPE pre-allocation.
\end{lemma}

Equivalently, $\bs{p}^*$ solves the optimization problem \eqref{eq:p_optimization}. 
\begin{proof}
	The proof will follow two sub-parts, 2-a and 2-b. In part 2-a, we first establish that $\bs{p}^*$ is a local maximizer of $\pi_A(\bs{p},R_A,R_B)$, which necessarily occurs when either $\mcal{B}_1=\mcal{B}$ or $\mcal{B}_2=\mcal{B}$. In part 2-b, we show that no choice of $\bs{p} \in \Delta_n(P)$ that results in both sets $\mcal{B}_1$ and $\mcal{B}_2$ being non-empty achieves a higher payoff than $\pi_A(\bs{p}^*,R_A,R_B)$, thus establishing Lemma \ref{lem:global_maximizer}.
	
	\noindent\textbf{Part 2-a:} $\bs{p}^*$ \emph{is a local maximizer of} $\pi_A(\bs{p},R_A,R_B)$.
	
	From Lemma \ref{lem:SOE} and the definition of $h_b(\kappa_A,\kappa_B)$, we find that the solution to \eqref{eq:SOE} under the pre-allocation $\bs{p}^*$ is always in one of two completely symmetric cases: 1) $\mcal{B}_1 = \mcal{B}$; or 2) $\mcal{B}_2 = \mcal{B}$.  Thus, we need to show $\bs{p}^*$ is a local maximizer in both cases.

    \smallskip\noindent\underline{{\bf Case 1 ($\mcal{B}_1 = \mcal{B}$):}} For $\bs{p} \in \Delta_n(P)$, the system \eqref{eq:SOE} is written
    \begin{equation}\label{eq:case22_SOE}
        \begin{aligned}
            &R_A = \sum_{b=1}^n \frac{(w_b\kappa_B - p_b)^2}{2w_b\kappa_B} \text{ and } R_B = \sum_{b=1}^n \frac{(w_b\kappa_B)^2-p_b^2}{2w_b\kappa_A} \\
            &\text{where } 0 < w_b\kappa_B - p_b \leq \kappa_A \text{ holds } \forall b \in \mcal{B}.
        \end{aligned}
    \end{equation}
    It yields the algebraic solution
    \begin{equation} \label{eq:case1solution_closedform}
        \begin{aligned}
            \kappa_B^* &= P+R_A + \sqrt{(P+R_A)^2 - \|\bs{p}\|_{\bs{w}}^2} \\ 
            \kappa_A^* &= \frac{(P+R_A)\kappa_B^* - \|\bs{p}\|_{\bs{w}}^2}{R_B}.
        \end{aligned}
    \end{equation}
    where $\|\bs{p}\|_{\bs{w}}^2 := \sum_{b=1}^n \frac{p_b^2}{w_b}$. This solution needs to satisfy the set of conditions $0 < w_b\kappa_B - p_b \leq \kappa_A \ \forall b \in \mcal{B}$, but the explicit characterization of these conditions is not needed to show that $\bs{p}^*$ is a local maximum. Indeed, first observe that the expression for $\kappa_B^*$ is required to be real-valued, which we can write as the condition
    \begin{equation}
        \bs{p} \in R^{(1n)} := \left\{\bs{p} \in \Delta_n(P) : \|\bs{p}\|_{\bs{w}}^2 < (P+R_A)^2 \right\}.
    \end{equation}
    We thus have a region $R^{(1n)}$ for which player $A$'s equilibrium payoff (Lemma \ref{lem:SOE}) is given by the expression
    \begin{equation}\label{eq:case22_payoff}
        \pi_A^{(1n)}(\bs{p}) := 1 - \frac{R_B}{f(\|\bs{p}\|_{\bs{w}})}\!\left(1 - \frac{\|\bs{p}\|_{\bs{w}}^2}{(P\!+\!R_A\!+\!f(\|\bs{p}\|_{\bs{w}}))^2} \right)
    \end{equation}
    where $f(\|\bs{p}\|_{\bs{w}}) := \sqrt{(P+R_A)^2 - \|\bs{p}\|_{\bs{w}}^2}$. The partial derivatives are calculated to be
    \begin{equation}
        \frac{\partial \pi_A^{(1n)}}{\partial p_b}(\bs{p}) = \frac{p_b}{w_b} \cdot\frac{2R_B }{f(\|\bs{p}\|_{\bs{w}})(P+R_A+f(\|\bs{p}\|_{\bs{w}}))^2}
    \end{equation}
    A critical point of $\pi_A^{(1n)}$ must satisfy $\bs{z}^\top \nabla \pi_A^{(1n)}(\bs{p})=0$ for any $\bs{z} \in \mbb{T}_n$, where we define $\mbb{T}_n := \{\bs{z} \in \mbb{R}^n : \sum_{b=1}^n z_b = 0\}$ as the tangent space of $\Delta_n(P)$. Indeed for any $\bs{p} \in R^{(1n)}$, we calculate
    \begin{equation}
        \begin{aligned}
            (\bs{p} -\bs{w}\cdot P)^\top \nabla \pi_A^{(1n)}(\bs{p}) &= g(\|\bs{p}\|_{\bs{w}})\cdot \left(\|\bs{p}\|_{\bs{w}}^2 - P^2 \right) \\
            &\geq 0
        \end{aligned}
    \end{equation}
    where $g(\|\bs{p}\|_{\bs{w}}):=\frac{2R_B }{f(\|\bs{p}\|_{\bs{w}})(P+R_A+f(\|\bs{p}\|_{\bs{w}}))^2} > 0$ for any $\bs{p} \in R^{(1n)}$. The inequality above is met with equality if and only if $\bs{p} = \bs{p}^*$. This is due to the fact that $\min_{\bs{p} \in \Delta_n(P)} \|\bs{p}\|_{\bs{w}}^2 = \|\bs{p}^*\|_{\bs{w}}^2 = P^2$. Thus, $\bs{p}^*$ is the unique maximizer of $\pi_A^{(1n)}(\bs{p})$ on $R^{(1n)}$.
    
    \smallskip\noindent\underline{{\bf Case 2 ($\mcal{B}_2 = \mcal{B}$):}}
    For $\bs{p}\in\Delta_n(P)$, the system is written as
    \[ R_A = \sum^n_{b=1} \frac{(w_b \kappa_A)^2}{2 w_b \kappa_B} \text{ and } 
       R_B = \sum^n_{b=1} \frac{(w_b \kappa_A-p_b)^2-(p_b)^2}{2 w_b \kappa_A}, \]
    where $w_b \kappa_B - p_b > w_b \kappa_A$ holds for all $b\in\mcal{B}$.  This readily yields the algebraic solution:
    \begin{equation} \label{eq:case2solution_closedform}
        \kappa_B^* = 2 \frac{( R_B-P)^2}{R_A} \text{ and } 
        \kappa_A^* = 2 (R_B-P).
    \end{equation}
    For this solution to be valid, the following conditions are required:
    
        \noindent $\bullet$ $\kappa_A^*,\kappa_B^* \in \mbb{R}_{++}$:  This requires that $R_B-P>0$.
        
        \noindent $\bullet$ $w_b \kappa_B^* - p_b > w_b \kappa_A^*$ for all $b\in\mcal{B}$:  This requires that
        \[ 2 \frac{(R_B-P)^2}{R_A} - 2 (R_B-P) - \max_{b} \{\frac{p_b}{w_b}\} > 0. \]
        The left-hand side is quadratic in $R_B-P$, and thus requires that either
        \[\begin{aligned}
            R_B-P < \frac{R_A}{2}\left(1 - \sqrt{1+\frac{2}{ R_A}\max_b\{\frac{p_b}{w_b}\}} \right)
        \end{aligned} \]
        or
        \begin{equation} \label{eq:case3_mostrestrictive}
            R_B-P > \frac{R_A}{2}\left(1 + \sqrt{1+\frac{2}{ R_A}\max_b\{\frac{p_b}{w_b}\}} \right).
        \end{equation}
        The former cannot hold since the numerator on the right-hand side is strictly negative, but $\kappa_A^*,\kappa_B^* \in \mbb{R}_{++}$ requires $ R_B-P>0$.  Thus, \eqref{eq:case3_mostrestrictive} must hold, Clearly, this is more restrictive than $R_B-P>0$. This dictates the boundary of Case 2.
    
    For any $\bs{p}\in\Delta_n(P)$ such that all battlefields are in Case 2, the expression for player $A$'s payoff in \eqref{eq:playerA_payoff} simplifies to
    \[ \pi_A(\bs{p},R_A,R_B) = \sum^n_{b=1} w_b \frac{\kappa_A^*}{2 \kappa_B^*} 
        = \frac{R_A}{2(R_B-P)}, \]
    where we use the expression for $\kappa_B^*$ and $\kappa_A^*$ in \eqref{eq:case2solution_closedform}.  Observe that player $A$'s payoff is constant in the quantity $\bs{p}$.  Thus, for any $\bs{p}$ that satisfies \eqref{eq:case3_mostrestrictive}, it holds that all battlefields are in Case 2, and that player $A$'s payoff is the above.  We conclude sub-part 2-a noting that, for given quantities $R_A$ and $P$, if there exists any $\bs{p}\in\Delta_n(P)$ such that \eqref{eq:case3_mostrestrictive} is satisfied, then $\bs{p}^* = \bs{w} \cdot P$ must also satisfy \eqref{eq:case3_mostrestrictive}, since $||\bs{p}||_\infty \geq ||\bs{p}^*||_\infty$ and the right-hand side in \eqref{eq:case3_mostrestrictive} is increasing in $||\bs{p}||_\infty$.
	
	\noindent\textbf{Part 2-b:} \emph{Any pre-allocation $\bs{p}$ that corresponds to a solution of \eqref{eq:SOE} with $\mcal{B}_1,\mcal{B}_2 \neq \varnothing$ satisfies $\pi_A(\bs{p},R_A,R_B) \leq \pi_A(\bs{p}^*,R_A,R_B)$.}
	
	For easier exposition, the proof of Part 2-b is presented in the Appendix. Together, Parts 2-a and 2-b imply that $\bs{p}^*$ is a global maximizer of the function $\pi_A(\bs{p},R_A,R_B)$, completing the proof of Lemma \ref{lem:global_maximizer}. 
\end{proof}

\noindent\textbf{\underline{Part 3}:} 
In the third and final part, we obtain the formulas for SPE payoffs reported in Theorem \ref{thm:equilibrium_characterization}.

\begin{proof}[Proof of Theorem \ref{thm:equilibrium_characterization}]
	We proceed to derive closed-form solutions for the SPE payoff $\pi_A^*(P,R_A,R_B)$. From Lemmas \ref{lem:SOE} and \ref{lem:global_maximizer}, the SPE payoff is attained by evaluating $\pi_A(\bs{p}^*,R_A,R_B)$, i.e. from equation \eqref{eq:playerA_payoff}. From the discussion of Part 2-a, this amounts to analyzing the two completely symmetric cases $\mcal{B}_1 = \mcal{B}$ and $\mcal{B}_2 = \mcal{B}$.

    \smallskip\noindent\underline{{\bf Case 1 ($\mcal{B}_1 = \mcal{B}$):}} Substituting $\bs{p}^*=\bs{w}\cdot P$ into \eqref{eq:case1solution_closedform} and simplifying, we obtain
    \begin{equation}\label{eq:kappa_case2}
        \begin{aligned}
            \kappa_B^* &= P+R_A + \sqrt{R_A(R_A + 2P)}  \\ 
            \kappa_A^* &= \frac{(P+R_A)\kappa_B^* - P^2}{R_B}.
        \end{aligned}
    \end{equation}
    Next, we verify that this solution satisfies the conditions $0 < \kappa_B^* - P \leq \kappa_A^*$ imposed by the case $\mcal{B}_1 = \mcal{B}$.
    
    \noindent{$\bullet$ $\kappa_B^* - P > 0$:} This holds by inspection.
    
    \noindent{$\bullet$ $\kappa_B^* - P \leq \kappa_A$:} We can write this condition as
    \begin{equation}\label{eq:case2_cond3}
            R_B - P \leq R_A + \frac{PR_A}{R_A + \sqrt{R_A(R_A + 2P)}}
    \end{equation}
    We note that whenever $R_B \leq P$, this condition is always satisfied. When $R_B  > P$, this condition does not automatically hold, and an equivalent expression of \eqref{eq:case2_cond3} is given by
    \begin{equation}
        R_A \geq \frac{2(R_B-P)^2}{P + 2(R_B-P)}.
    \end{equation}
    Observe that $R_A = \frac{2(R_B-P)^2}{P + (R_B-P)}$ satisfies \eqref{eq:case2_cond3} with equality, and is in fact the only real solution (one can reduce it to a cubic polynomial in $R_A$). 
    
    When these conditions hold, the equilibrium payoff $\pi_A^*(P,R_A,R_B) = \pi_A(\bs{p}^*,R_A,R_B)$ can be directly computed from Lemma \ref{lem:SOE}, i.e. \eqref{eq:playerA_payoff}. It is  given by the expression \eqref{eq:allcase1_payoff}.
    
    
    \smallskip\noindent\underline{{\bf Case 2 ($\mcal{B}_2 = \mcal{B}$):}} Substituting $\bs{p}= \bs{w}\cdot P$ into \eqref{eq:case2solution_closedform} and simplifying, we obtain
    \begin{equation}
        \kappa_A^* = \frac{2(R_B - P)}{W} \text{ and } \kappa_B^* = \frac{2(R_B - P)^2}{R_A}.
    \end{equation}
    The solution satisfies the conditions $0 < \kappa_A^* < \kappa_B^* - P$ imposed by the case $\mcal{B}_2 = \mcal{B}$ if and only if $R_B > P$ and  $R_A > \frac{2(R_B-P)^2}{P + (R_B-P)}$. When this holds, the SPE payoff is calculated from \eqref{eq:playerA_payoff} to be $\pi_A^*(P,R_A,R_B) = W\cdot \frac{R_A}{2(R_B-P)}$. 
\end{proof}


\section{Interplay between resource types}\label{sec:interplay_results}

\begin{figure*}[t]
    \centering
    \includegraphics[scale=0.3]{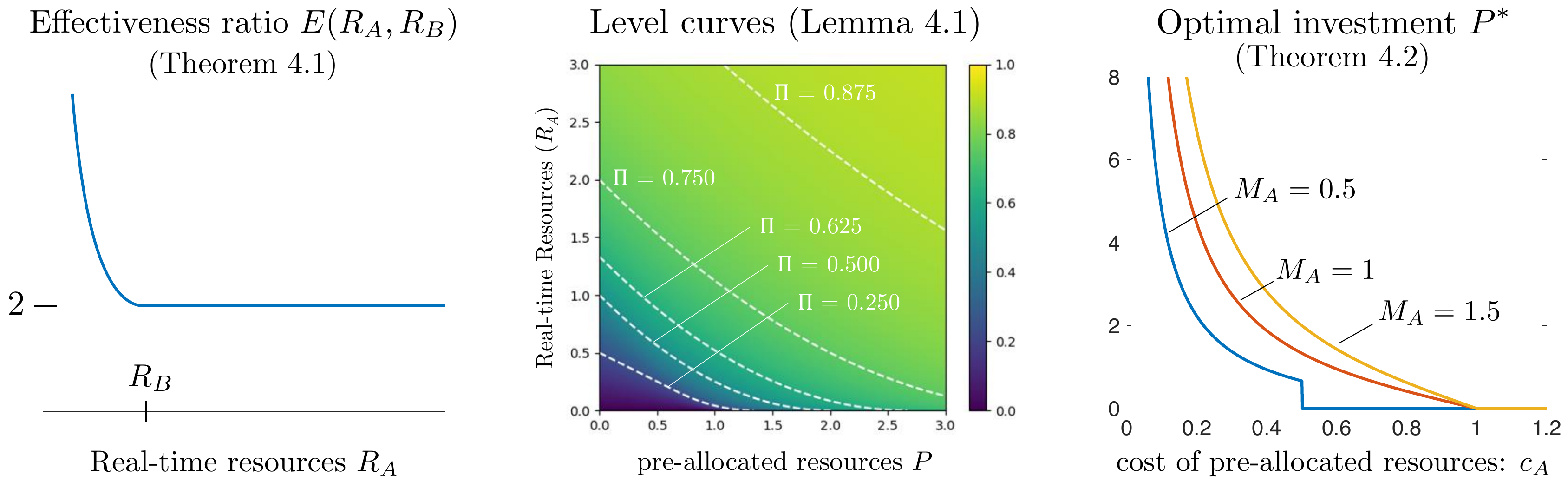}
    \caption{(Left) A plot of the effectiveness ratio $E(R_A,R_B)$ (Theorem \ref{thm:ratio}), which quantifies the multiplicative factor of pre-allocated resources needed to achieve the same performance as an amount of real-time resources $R_A$. Notably, real-time resources are at least twice as effective as an equivalent amount of pre-allocated resources. (Center) This plot shows a collection of level curves for player $A$'s SPE payoff. A level curve corresponds to a fixed performance level $\Pi$, and any point $(P,R_A)$ on the level curve satisfies $\pi_A^*(P,R_A,R_B) = \Pi$ (Lemma \ref{lem:level_set}).  (Right) This plot shows player $A$'s optimal investment in pre-allocated resources $P^*$ when it has a per-unit cost of $c_A$ and a fixed monetary budget of $M_A$ to invest in both types of resources (Theorem \ref{thm:investment}). Player $A$ invests the remaining $M_A - c_A P^*$ in real-time resources. In these plots, we set $R_B = 1$, and $W = 1$.}
    \label{fig:interplay_figure}
\end{figure*}

In this section, we present some implications from Theorem \ref{thm:equilibrium_characterization} regarding the interplay between pre-allocated and real-time resources. Specifically, we seek to compare the relative effectiveness of the two types of resources in the GL-P game by first quantifying an \emph{effectiveness ratio}. We then leverage this analysis to address how player $A$ should invest in both types of resources when they are costly to acquire.

\subsection{The effectiveness ratio}

We define the effectiveness ratio as follows. 

\begin{definition}
	For a given $R_A,R_B > 0$, let $P^\text{eq}(R_A,R_B) > 0$ be the unique value such that $\pi_A^*(P^\text{eq},0,R_B) = \pi_A^*(0,R_A,R_B)$. The \emph{effectiveness ratio} is defined as
	\begin{equation}
		E(R_A,R_B) := \frac{P^\text{eq}}{R_A}
	\end{equation}
\end{definition}

In words, $P^\text{eq}$ is the amount of pre-allocated resources required to achieve the same level of performance  as the amount $R_A$ of real-time resources  in the absence of pre-allocated resources. The effectiveness ratio $E$ thus quantifies the multiplicative factor of pre-allocated resources needed compared to real-time resources $R_A$.

The following result establishes the effectiveness ratio for any given parameters.

\begin{theorem}\label{thm:ratio}
	For a given $R_A,R_B > 0$, the effectiveness ratio is
	\begin{equation}
		E(R_A,R_B) = 
		\begin{cases}
            		2 & \quad \text{if } R_A \geq R_B, \\
            		\frac{2(R_B)^2}{R_A(2R_B-R_A)} & \quad \text{if } R_A < R_B.
        		\end{cases}
	\end{equation}
\end{theorem}

Here, it is interesting to note that the ratio $E$ is lower-bounded by $2$ -- real-time resources are at least twice as effective as pre-allocated resources. Additionally, as $R_A \to 0^+$, the ratio grows unboundedly $E \to \infty$. This is due to the fact that without any real-time resources, player $A$ needs $P \geq R_B$ pre-allocated resources to obtain a positive payoff (see third case of Theorem \ref{thm:equilibrium_characterization}). A plot of the ratio $E$ is shown in the left Figure \ref{fig:interplay_figure}.

The proof of Theorem \ref{thm:ratio} relies on the following technical lemma, which provides the level curves of the SPE payoff $\pi_A^*(P,R_A,R_B)$. A level curve with fixed performance level $\Pi \in [0,1]$ is defined as the set of points
 \begin{equation}
    	L_\Pi := \{ (P,R_A) \in \mbb{R}_+^2 : \pi_A^*(P,R_A,R_B) = \Pi \}.
\end{equation}

\begin{lemma}\label{lem:level_set}
    Given any $R_B>0$ and $\bs{w}\in\mbb{R}^n_{++}$, fix a desired performance level $\Pi \in [0,1]$.  The level curve $L_\Pi$ is given by
     \begin{equation}
    	L_\Pi = \bigcup_{P \in \left[0,\frac{R_B}{1-\Pi} \right]} \left(P,R_\Pi(P) \right)
    \end{equation}
    where if $0\leq \Pi < \frac{1}{2}$,
    \begin{equation} \label{eq:level_curve_1and2}
        R_\Pi(P) = \!\!\begin{cases}
            2\Pi (R_B-P) &\, \text{for } P \in \Big[ 0, \frac{(1-2\Pi)R_B}{1-\Pi} \Big) \\
            \!\!\frac{(R_B-(1-\Pi)P)^2}{2R_B(1-\Pi)} &\, \text{for } P \in \Big[\frac{(1-2\Pi)R_B}{1-\Pi},\frac{WR_B}{1-\Pi}\Big]
        \end{cases}
    \end{equation}
    and if $\frac{1}{2} \leq \Pi \leq 1$, 
    \begin{equation} \label{eq:level_curve_3}
        R_\Pi(P) = \frac{(R_B-(1-\Pi)P)^2}{2R_B(1-\Pi)}
    \end{equation}
    If $P> \frac{R_B}{1-\Pi}$, then $\pi^*_A(P,R_A,R_B)>\Pi$ for any $R_A\geq0$.
\end{lemma}
The proof of this Lemma directly follows from the expressions in Theorem \ref{thm:equilibrium_characterization} and is thus omitted. In the center Figure \ref{fig:interplay_figure}, we illustrate level curves associated with varying performance levels $\Pi$. We can now leverage the above Lemma to complete the proof of Theorem  \ref{thm:ratio}.

\begin{proof}[Proof of Theorem \ref{thm:ratio}]
	First, suppose $R_A < R_B$. Then $\pi_A^*(0,R_A,R_B) = \frac{R_A}{2R_B} < 1/2$. Focusing on the level curve associated with the value $\Pi = \frac{R_A}{2R_B}$, the quantity $P^\text{eq}$ is determined as the endpoint of this curve where there are zero real-time resources. From \eqref{eq:level_curve_1and2} , this occurs when $P = \frac{R_B}{1-\Pi} = \frac{2(R_B)^2}{2R_B - R_A}$.
	
	Now, suppose $R_A \geq R_B$. Then $\pi_A^*(0,R_A,R_B) = (1-\frac{R_B}{2R_A} \geq 1/2$. Similarly, the quantity $P^\text{eq}$ is determined as the endpoint of the level curve associated with $\Pi = (1-\frac{R_B}{2R_A})$. From \eqref{eq:level_curve_3} , this occurs when $P = \frac{R_B}{1-\Pi} = 2R_A$. 
	
\end{proof}

\subsection{Optimal investment in resources}

In addition to the effectiveness ratio, the interplay between the two types of resources is also highlighted by the following scenario: player $A$ has an opportunity to make an investment decision regarding its resource endowments. That is, the pair $(P,R_A)\in\mbb{R}^2_+$ is a strategic choice made by player $A$ before the game $\text{GL-P}(P,R_A,R_B,\bs{w})$ is played. Given a monetary budget $M_A > 0$ for player $A$, any pair $(P,R_A)$ must belong to the following set of feasible investments:
\begin{equation}\label{eq:linear_cost_constraint}
        \mcal{I}(M_A) := \{(P,R_A) : R_A + c_A P \leq M_A\}
\end{equation}
where $c_A\geq 0$ is the per-unit cost for purchasing pre-allocated resources, and we assume the per-unit cost for purchasing real-time resources is 1 without loss of generality. We are interested in characterizing player $A$'s optimal investment subject to the above cost constraint, and given player $B$'s resource endowment $R_B>0$. This is formulated as the following optimization problem:
\begin{equation}\label{eq:optimal_investment_problem}
    \pi_A^\mathrm{opt} := \max_{(P,R_A) \in \mcal{I}(M_A)} \pi_A^*(P,R_A,R_B).
\end{equation}




In the result below, we derive the complete solution to the optimal investment problem \eqref{eq:optimal_investment_problem}.

\begin{theorem}\label{thm:investment}
    Fix a monetary budget $M_A>0$, relative per-unit cost $c_A>0$, and $R_B>0$ real-time resources for player $B$.  Then, player $A$'s optimal investment in pre-allocated resources in \eqref{eq:optimal_investment_problem}  is
    \begin{equation}\label{eq:opt_p}
        P^* = 
        \begin{cases}
            \frac{2(1-c_A)}{2-c_A}\frac{M_A}{c_A}, &\text{if } c_A < t \\
            \in [0,\frac{2(1-c_A)}{2-c_A}\frac{M_A}{c_A}], &\text{if } c_A = t \\
            0, &\text{if } c_A > t
        \end{cases}.
    \end{equation}
    where $t := \min\{1,\frac{M_A}{R_B}\}$. The optimal investment in real-time resources is $R_A^* = M_A - c_A P^*$. The resulting payoff $\pi_A^\mathrm{opt}$ to player $A$ is given by
    \begin{equation}
        \begin{cases}
            1 - \frac{R_B}{2M_A}c_A(2-c_A), &\text{if } c_A < t \\
            1 - \frac{R_B}{2M_A}, &\text{if } c_A \geq t \text{ and } \frac{M_A}{R_B} \geq 1 \\
            \frac{M_A}{2R_B}, &\text{if } c_A \geq t \text{ and } \frac{M_A}{R_B} < 1 
        \end{cases}.
    \end{equation}
\end{theorem}

A plot of the optimal investment $P^*$ \eqref{eq:opt_p} is shown in the right Figure \ref{fig:interplay_figure}. If the cost $c_A$ exceeds 1, then there is no investment in pre-allocated resources since they are less effective than real-time resources. Thus, $c_A$ must necessarily be cheaper than real-time resources in order to invest in any positive amount. We note that while an optimal investment can purely consist of real-time resources, no optimal investment from Theorem \ref{thm:investment} can purely consist of pre-allocated resources. Interestingly, when the monetary budget is small ($M_A < 1$), there is a discontinuity in the investment level $P^*$ at $c_A = R_B$. A visual illustration of how the optimal investments are determined is shown in Figure \ref{fig:investment_method}, which is detailed in the proof below. 

\begin{figure}
    \centering
    \includegraphics[scale=0.35]{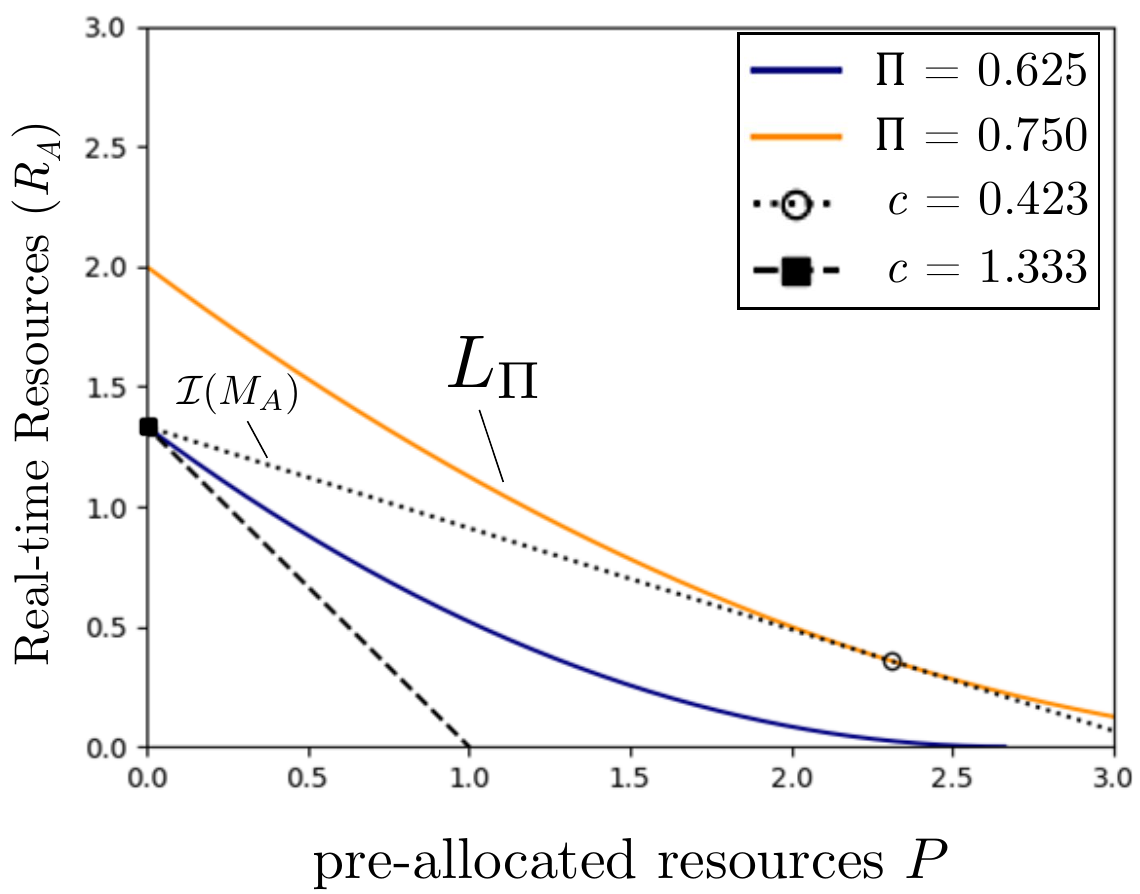}
    \caption{This plot illsutrates how to determine the optimal investment $(P^*,R^*_A)\in\mbb{R}^2_+$ subject to the  cost constraint in \eqref{eq:linear_cost_constraint}.  The set of feasible investments $\mcal{I}(M_A)$ is the line segment connecting $(0,M_A)$ and $(M_A/c_A,0)$. The optimal investment lies on the level curve tangent to this line segment. For example, when $c=0.423$, the optimal investment is $(2.309,0.357)$ (unfilled circle), which gives a performance level of $\Pi = 0.75$. For sufficiently high cost $c_A$, $\mcal{I}(M_A)$ will not be tangent to any level curve, and the optimal investment is $(0,M_A)$. For example, when $c_A=1.333$, the highest level curve that intersects $\mcal{I}(M_A)$ is $\Pi = 0.625$, and the optimal investment is $(0,4/3)$ (filled square).}\label{fig:investment_method}
\end{figure}

\begin{proof}
    We first observe that for any $\Pi \in (0,1)$, the level curve $R_\Pi(P)$ (from Lemma \ref{lem:level_set}) is strictly decreasing and convex in $P \in [0,\frac{R_B}{1-\Pi}]$. Hence, the function $\pi_A(P,R_A,R_B)$ is quasi-concave in $(P,R_A)$.  Observe that the set of points $(P,R_A)\in\mbb{R}^2_+$ that satisfy $R_A + c_A P = M_A$ consists of the line segment $R_A=M_A- c_A P$, $P\in[0,M_A/c_A]$, with slope $-c_A$, and end-points $(M_A,0)$ and $(0,M_A/c_A)$.  Thus, the optimization amounts to finding the highest level curve that intersects with $R_A=M_A-c_A P$, $P\in[0,M_A/c_A]$.  
    
    The slope of a level curve $R_\Pi(P)$ at $P=0$ is
    \begin{equation}
        \frac{\partial R_\Pi}{\partial P}(0) =
        \begin{cases}
            -2\Pi, &\text{if } \Pi < \frac{1}{2} \\
            -1, &\text{if } \Pi \geq \frac{1}{2}
        \end{cases}.
    \end{equation}
    Let $M_A\geq0$ such that $\pi^*_A(0,M_A,R_B)=\Pi\geq 1/2$.  Then, note that, if $-c_A<-1$ (or, equivalently $c_A>1$), then $R_A=M_A-c_A P$ shrinks faster in $P$ than the level curve $R_\Pi(P)$ by monotonicity and convexity of the level curve in $P$.  Thus, the allocation $(P,R_A)=(0,M_A)$ maximizes $A$'s payoff, as all other points on the line segment $R_A=M_A-c_AP$, $P\in[0,M_A/c_A]$, intersect with strictly lower level curves.  Similarly, for $M_A\geq0$ such that $\pi^*_A(0,M_A,R_B)=\Pi<1/2$, the allocation $(P,R_A)=(0,M_A)$ maximizes $A$'s payoff when $c_A>2\Pi$.  Since the condition $\pi^*_A(0,M_A,R_B)=\Pi\geq 1/2$ is equivalent to $M_A\geq R_B$ and $\pi^*_A(0,M_A,R_B)=\frac{R_A}{2 R_B}$ when $M_A<R_B$, it follows that $P^*=0$ if $c_A>\min\{1,\frac{M_A}{R_B}\}$.  For the remainder of the proof, we use $t=1$ (resp. $t= \frac{M_A}{R_B}$) and $M_A\geq R_B$ (resp. $M_A<R_B$) interchangeably.
    
    Suppose that $c_A = t$ and $t = \frac{M_A}{R_B}$. Then the level curve corresponding to $\Pi(M_A) = \frac{W}{2}c_A$ has an interval of budget-feasible points $(P,R_A)$ parameterized by $P \in [0,(1-\frac{c_A}{2-c_A})\frac{M_A}{c_A}]$ with $R_A = M_A - cP$. If $t = 1$, then there is a single budget-feasible point $(P,R_A) = (0,M_A)$ for the level curve corresponding to $\Pi(M_A) = (1-\frac{R_B}{2M_A})$. In both cases, there are no budget-feasible points for any level curve corresponding to $\Pi > \Pi(M_A)$. 
    
    Now, suppose $t = 1$ and $c_A < t$. We wish to find the level curve for which the line $(P,M_A-c_A P)$, $P \in [0,M_A/c_A]$, lies tangent. The point(s) of tangency yields the optimal solution due to the quasi-concavity of $\pi_A^*$. Furthermore, since $M_A\geq R_B$ and $c_A < 1$, a solution $(P,R_A)$ must satisfy $\Pi \in [\frac{1}{2},1]$ and
    \begin{equation}\label{eq:tangent_point}
        \begin{aligned}
            \frac{\partial R_\Pi}{\partial P}(P^*) &= \frac{P^*(1-\Pi)}{R_BW} - 1 &&= -c_A \\
            R_\Pi(P^*) &= \frac{(R_B - (1-\Pi)P^*)^2}{2R_B(1-\Pi)} &&= M_A - c_A P^*
        \end{aligned}
    \end{equation}
    
    \noindent From the first equation, we obtain $P^* = \frac{R_B(1-c_A)}{1-\Pi}$. Plugging this expression into the second equation, we obtain $\Pi = (1 - \frac{R_B}{2M_A}c_A(2-c_A)) \in [\frac{1}{2},1]$, which leads to the unique solution $P^* = (1-\frac{c_A}{2-c_A})\frac{M_A}{c_A} \leq \frac{M_A}{c_A}$.
    
    Lastly, suppose $c_A < t$ and $t = \frac{M_A}{R_B}$ ($M_A < R_B$). Similar to the preceding case, we seek the highest level curve for which the budget constraint is tangent. Due to the assumption that $M_A < R_B$ and $c < \frac{M_A}{R_B}$, we observe that tangent points cannot exist for $P < \frac{1-2\Pi}{1-\Pi}R_B$ and $\Pi < \frac{1}{2}$, i.e. in the region where the level curve is linear. Thus, it must be that either $\Pi < \frac{1}{2}$ and $P \geq \frac{1-2\Pi}{1-\Pi}R_B$, or $\Pi \geq \frac{1}{2}$. In either case, a solution must also satisfy the equations in \eqref{eq:tangent_point}, from which we obtain an identical expression for $P^*$. 
\end{proof}


\section{Two-sided pre-allocations}\label{sec:stackelberg}

The scenarios studied thus far have considered one-sided pre-allocations, where only player $A$ has the opportunity for early investments. The goal in this section is to take preliminary steps in understanding how multiple rounds of early investments, on the part by both competitors, impacts the players' performance in the final stage. We will consider a scenario where player $B$ has an opportunity to respond to the pre-allocation decision of player $A$ with its own pre-allocated resources, which we formulate as a Stackelberg game. 

\begin{remark}
    Before formalizing this game, we remark that such a scenario admits positive and negative pre-allocations, i.e. $p_b > 0$ for some subset of battlefields and $p_b < 0$ on the others. Here, $p_b<0$ means that the amount $|p_b|$ of pre-allocated resources favors player $B$.  While the work in \cite{Vu_EC2021} establishes existence of equilibrium for any such pre-allocations as well as numerical approaches to compute equilibria to arbitrary precision, it does not provide analytical characterizations of them. Indeed, while our current techniques (i.e. from Theorem \ref{thm:equilibrium_characterization}) analytically derive the equilibria for any positive pre-allocation vector, they are yet unable to account for such two-sided favoritism. Developing appropriate methods is subject to future study.
\end{remark}

In light of the aforementioned limitations, we may still investigate the impact of player $B$'s response in the context of a single-battlefield environment\footnote{In contrast to Colonel Blotto games, General Lotto games with a single battlefield still provides rich insights that often generalize to multi-battlefield scenarios \cite{Hart_2008,Hart_2016,Paarporn_2021_budget}.}. The Stackelberg game is defined as follows. Player $A$ has a monetary budget $M_A$ with per-unit cost $c_A \in (0,1)$ for stationary resources. Similarly, player $B$ has a monetary budget $M_B$ with per-unit cost $c_B \in (0,1)$. The players compete over a single battlefield of unit value.

\smallskip \noindent -- \emph{Stage 1:} Player $A$ chooses its pre-allocation investment $p_A \in [0,\frac{M_A}{c_A}]$. This becomes common knowledge.

\smallskip \noindent -- \emph{Stage 2:} Player $B$ chooses its pre-allocation investment $p_B \in [0,\frac{M_B}{c_B}]$. 

\smallskip \noindent -- \emph{Stage 3:} The players engage in the General Lotto game with favoritism $\text{GL-F}(p_A-p_B,M_A-c_Ap_A,M_B-c_Bp_B)$.  Players derive the final payoffs 
\begin{equation}
    \begin{aligned}
        u_A(p_A,p_B) &:= \pi_A^*(p_A-p_B,M_A-c_Ap_A,M_B-c_Bp_B) \\
        u_B(p_A,p_B) &:= 1 - u_A(p_A,p_B)
    \end{aligned}
\end{equation}
Note that $p_A - p_B$ is the favoritism to player $A$. When it is non-negative, $\pi_A^*$ is given precisely by Theorem \ref{thm:equilibrium_characterization}. When it is negative, $\pi^*_A = 1 - \pi^*_B$ where $\pi_B^*$ is given as in Theorem \ref{thm:equilibrium_characterization} with the indices switched.

Let us denote this game as $\text{GL-S}(\{M_i,c_i\}_{i=A,B})$. We seek to characterize the following equilibrium concept.

\begin{definition}
    The investment profile $(p_A^*,p_B^*)$ is a \emph{Stackelberg equilibrium} if 
    \begin{equation}\label{eq:A_stackelberg}
        p_A^* \in \arg\max_{p_A\in[0,M_A/c_A]} \left(\min_{p_B \in [0,M_B/c_B]} u_A(p_A,p_B) \right)
    \end{equation}
    and 
    \begin{equation}
        p_B^* \in \arg\min_{p_B\in[0,M_B/c_B]} u_A(p_A^*,p_B).
    \end{equation}
\end{definition}

Note that the definition in \eqref{eq:A_stackelberg} is in a max-min form, since the final payoffs in GL-S are constant-sum. The characterization of the Stackelberg equilibrium is given in the result below.

\begin{proposition}\label{prop:stack_equil}
    The Stackelberg equilibrium of $\text{GL-S}(\{M_i,c_i\}_{i=A,B})$ is given as follows.
    \begin{enumerate}[leftmargin=*]
        \item Suppose $\frac{M_B}{c_B} \leq M_A$. Then $p_A^*$ is given according to Theorem \ref{thm:investment} and $p_B^* = 0$.
        \item Suppose $M_A < \frac{M_B}{c_B} \leq \frac{M_A}{c_A}$. If $p_A^\dagger < \frac{2(1-c_A)}{2-c_A}\frac{M_A}{c_A}$, then $p_A^*$ is given according to Theorem \ref{thm:investment} and $p_B^* = 0$, where $p_A^\dagger \in (0,\frac{M_A}{c_A}]$ is the unique value that satisfies $u_B(p_A,0) = u_B(p_A,\hat{p}_B)$, with 
        \begin{equation}\label{eq:pBcrit}
            \hat{p}_B := \frac{M_B}{c_B} - \frac{M_B-c_Bp_A}{2-c_B}.
        \end{equation} 
        If $p_A^\dagger \geq \frac{2(1-c_A)}{2-c_A}\frac{M_A}{c_A}$, then $p_A^* = p_A^\dagger$, and $p_B^* = 0$ or $\hat{p}_B$. 
        \item Suppose $\frac{M_A}{c_A} < \frac{M_B}{c_B}$. Then $p_A^* = 0$ and $p_B^* = \hat{p}_B$.
    \end{enumerate}
\end{proposition}

Several comments are in order. In the first interval $\frac{M_B}{c_B} \leq M_A$, player $B$ is sufficiently weak such that it does not respond with any of its own stationary resources against any player $A$ investment $p_A \in [0,\frac{M_A}{c_A}]$. Thus, the Stackelberg solution recovers the result from Theorem \ref{thm:investment}. For a large part of the middle interval $M_A < \frac{M_B}{c_B} \leq \frac{M_A}{c_A}$, $p_A^*$ coincides with the investment from Theorem \ref{thm:investment}, which forces player $B$ to respond with zero stationary resources (i.e. when $p_A^\dagger < \frac{2(1-c_A)}{2-c_A}\frac{M_A}{c_A}$). However, when $p_A^\dagger \geq \frac{2(1-c_A)}{2-c_A}\frac{M_A}{c_A}$, player $A$'s optimal investment makes player $B$ indifferent between responding with zero or with an amount $\hat{p}_B > p_A^*$ of stationary resources that exceeds the pre-allocation of player $A$. In the last interval $\frac{M_A}{c_A} < \frac{M_B}{c_B}$, player $B$ is sufficiently strong such that it is optimal for player $A$ to invest zero stationary resources. 

\subsection{The impact of responding}

We illustrate the implications of Proposition \ref{prop:stack_equil} regarding the responding player's ($B$) performance in the numerical example shown in Figure \ref{fig:stackelberg_payoffs}. Here, we compare player $B$'s Stackelberg equilibrium payoff to the payoff it would have obtained if it did not have an opportunity to respond. We recall that this payoff was characterized in Theorem \ref{thm:investment}. By definition, the Stackelberg payoff is necessarily at least as high as the non-responding payoff -- one performs better being able to respond to a pre-allocation. However, what is notable in Figure \ref{fig:stackelberg_payoffs} is a significant, discontinuous increase in  payoff for player $B$ once its budget $\frac{M_B}{c_B}$ is sufficiently high (i.e. exceeds $\frac{M_A}{c_A}$). 

The presence of the discontinuity strongly suggests that being in a resource-advantaged position with regards to the monetary budget $\frac{M_B}{c_B}$ is a crucial factor for performance in multi-stage resource allocation. At the same time, Proposition \ref{prop:stack_equil} asserts that no response is optimal if player $B$ is resource-disadvantaged (first item). This conclusion is contrasted with the classic simultaneous-move General Lotto games, in which there is no such discontinuity in the equilibrium payoffs -- they vary continuously in the players' budgets \cite{Hart_2008,Kovenock_2020}.

\begin{figure}[h]
    \centering
    \includegraphics[scale=0.35]{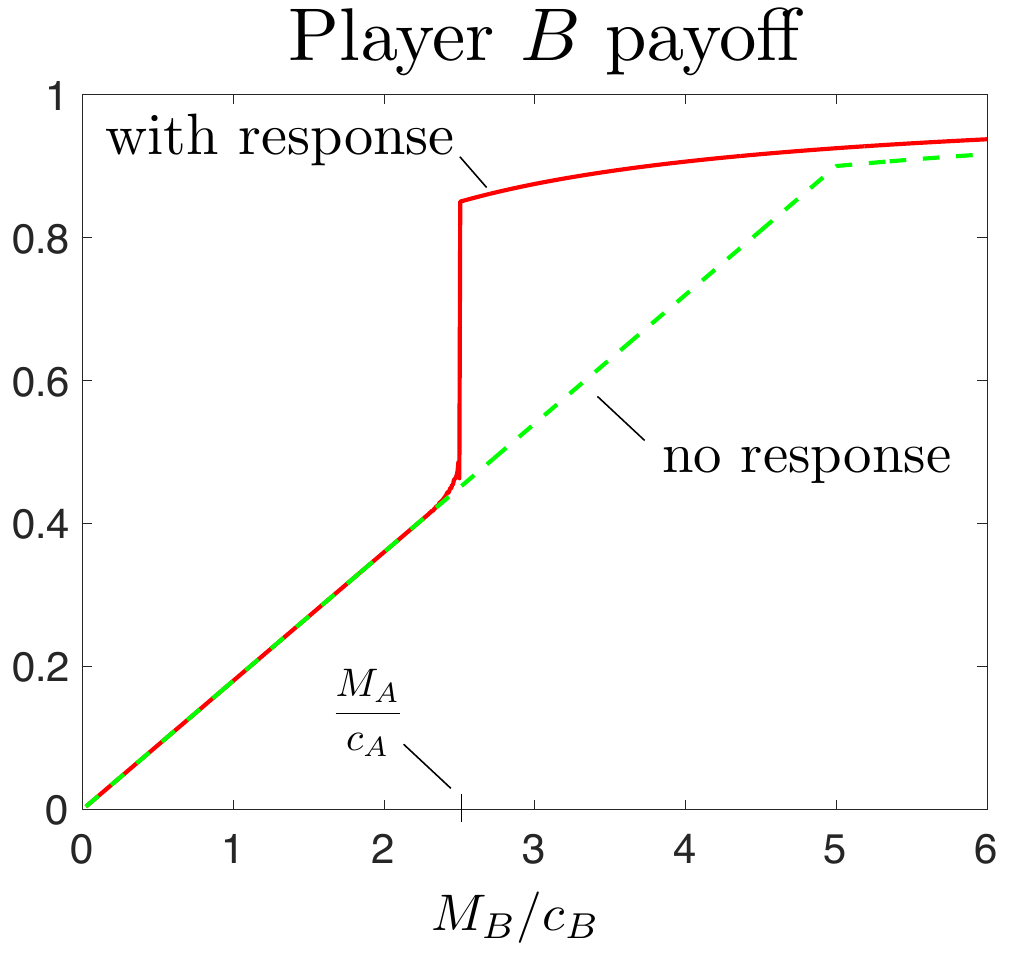}
    \caption{This plot illustrates the Stackelberg equilibrium payoff (red line, Proposition \ref{prop:stack_equil}) to player $B$ contrasted with its payoff if it did not have the opportunity to respond with pre-allocated resources, i.e. setting $p_B = 0$ (green dashed line, Theorem \ref{thm:investment}). Notably, there is a dramatic improvement in performance when player $B$ is sufficiently budget-rich, $\frac{M_B}{c_B} = \frac{M_A}{c_A}$. In this example, we set $M_A = 0.5$, $c_A = 0.2$, and $c_B = 0.5$. We vary $M_B$ from $0$ to $3$.}
    \label{fig:stackelberg_payoffs}
\end{figure}

The remainder of this section provides the proof of Proposition \ref{prop:stack_equil}, which utilizes two supporting Lemmas.

\subsection{Follower's best-response}

We begin by analyzing player $B$'s best-response to any player $A$ pre-allocation $p_A$. We need to find $p_B^*$ that solves
\begin{equation}
    \max_{p_B\in[0,M_B/c_B]} u_B(p_A,p_B).
\end{equation}

Let us denote $\tbfp = (p_A,p_B)$, $R_A = M_A - c_Ap_A$, and $R_B = M_B - c_Bp_B$. Define $f_A(\tbfp) := R_A + \sqrt{R_A(R_A + 2(p_A-p_B))}$, $f_B(\tbfp) := \sqrt{R_B}$, and $g_B(\tbfp) := \sqrt{R_B + 2(p_B - p_A)}$. Define
\begin{equation}
    \begin{aligned}
        u_B^{1A}(\tbfp) &:= \frac{M_B - c_Bp_B}{2R_A}\left(\frac{f_A(\tbfp) }{p_A-p_B + f_A(\tbfp)} \right)^2 \\
        u_B^{2A}(\tbfp) &:= 1 - \frac{R_A}{2(M_B - c_Bp_B - (p_A-p_B) )} \\
        u_B^{1B}(\tbfp) &:= 1 - \frac{R_A}{2}\left(\frac{f_B(\tbfp) + g_B(\tbfp)}{p_B - p_A + f_B(\tbfp)(f_B(\tbfp) + g_B(\tbfp))} \right)^2 \\
        u_B^{2B}(\tbfp) &:= \frac{M_B - c_Bp_B}{2(R_A - (p_B - p_A))}
    \end{aligned}
\end{equation}

From Theorem \ref{thm:equilibrium_characterization}, one can write player $B$'s payoff in GL-S as
\begin{equation}
    u_B(\tbfp) =
    \begin{cases}
        u_B^{1A}(\tbfp) &\text{if } \tbfp \in \mc{R}^{1A} \\
        u_B^{2A}(\tbfp) &\text{if } \tbfp \in \mc{R}^{2A} \\
        u_B^{1B}(\tbfp) &\text{if } \tbfp \in \mc{R}^{1B} \\
        u_B^{2B}(\tbfp) &\text{if } \tbfp \in \mc{R}^{2B}
    \end{cases}
\end{equation}
where
\begin{equation}
    \begin{aligned}
        \mc{R}^{1A} &= \{p_A \geq p_B : R_B < p_A - p_B, \text{ or }  \\
        & R_B \geq p_A - p_B \text{ and } R_A \geq \frac{2(R_B - (p_A-p_B))^2}{2R_B - (p_A-p_B) }\} \\
        \mc{R}^{2A} &= \{p_A \geq p_B \} \backslash \mc{R}^{1A}  \\
        \mc{R}^{1B} &= \{p_B > p_A : R_A < p_B - p_A, \text{ or }  \\
        & R_A \geq p_B - p_A \text{ and } R_B \geq \frac{2(R_A - (p_B-p_A))^2}{2R_A - (p_B-p_A) }\} \\
        \mc{R}^{2B} &= \{p_B > p_A \} \backslash \mc{R}^{1B}
    \end{aligned}
\end{equation}
The following Lemma details player $B$'s payoff for any response $p_B \in [0,\frac{M_B}{c_B}]$.
\begin{lemma}\label{lem:uB_pB}
    Consider any fixed strategy $(p_A,R_A)$ for player $A$. Player $B$'s payoff is given as follows.
    \begin{enumerate}[label=\alph*)]
        \item If $\frac{M_B}{c_B} \leq R_A + p_A$, then $u_B(\tbfp)$ is decreasing for all $p_B \in [0,\frac{M_B}{c_B}]$.
        \item  If $R_A + p_A < \frac{M_B}{c_B} \leq \frac{R_A}{c_B} + p_A$, then
        \begin{equation}
            u_B(\tbfp) = 
            \begin{cases}
                u_B^{1A}(\tbfp), &\text{if } p_B \in [0,p_A] \\
                u_B^{2B}(\tbfp), &\text{if } p_B \in (p_A,p_B^{1B}] \\
                u_B^{1B}(\tbfp), &\text{if } p_B \in (p_B^{1B},\frac{M_B}{c_B}]
            \end{cases}
        \end{equation}
        where $p_B^{1B}$ is the unique solution to
        \begin{equation}
             F(p_B) := \frac{(R_A+p_A - p_B)^2}{2R_A+p_A - p_B} = M_B-c_Bp_B.
        \end{equation}
        \item If $\frac{R_A}{c_B} + p_A < \frac{M_B}{c_B} \leq \frac{1}{c_B}\left(p_A + \frac{R_A + \sqrt{R_A(R_A + 2 p_A)}}{2}\right)$, then
        \begin{equation}
            u_B(\tbfp) =
            \begin{cases}
                u_B^{1A}(\tbfp), &p_B \in [0,p_B^{1A}] \\
                u_B^{2A}(\tbfp), &p_B \in (p_B^{1A},p_A] \\
                u_B^{1B}(\tbfp), &p_B \in (p_A,\frac{M_B}{c_B}]
            \end{cases}
        \end{equation}
        where $p_B^{1A}$ is the unique solution to
        \begin{equation}
            G(p_B) := \frac{2(M_B+(1-c_B)p_B - p_A)^2}{2(M_B - c_Bp_B)-(p_A-p_B)} = R_A.
        \end{equation}
        on $p_B \in (0,p_A)$
        \item If $\frac{M_B}{c_B} > \frac{1}{c_B}\left(p_A + \frac{R_A + \sqrt{R_A(R_A + 2 p_A)}}{2} \right)$, then
        \begin{equation}
            u_B(\tbfp) = 
            \begin{cases}
                u_B^{2A}(\tbfp), &\text{if } p_B \in [0,p_A] \\
                u_B^{1B}(\tbfp), &\text{if } p_B \in (p_A,\frac{M_B}{c_B}].
            \end{cases}
        \end{equation}
    \end{enumerate}
\end{lemma}

The proof is deferred to Appendix \ref{sec:follower_lemma}. The final Lemma characterizes player $B$'s best response to any fixed player $A$ strategy.

\begin{lemma}\label{lem:BR_B}
    Consider any fixed strategy $(p_A,R_A)$ for player $A$. Player $B$'s best-response
    \begin{equation}
        p_B^* := \argmax{p_B \in [0,\frac{M_B}{c_B}]} u_B(p_A,p_B)
    \end{equation}
    is determined according to
    \begin{equation}
        p_B^* = 
        \begin{cases}
            0, &\text{if } \frac{M_B}{c_B} < h(R_A,p_A) \\
            \hat{p}_B, &\text{if } \frac{M_B}{c_B} > h(R_A,p_A) \\
            0 \text{ or } \hat{p}_B, &\text{if } \frac{M_B}{c_B} = h(R_A,p_A)
        \end{cases}
    \end{equation}
    where $\hat{p}_B$ was defined in \eqref{eq:pBcrit}, and $h(R_A,p_A) \in (R_A+p_A,p_A+\frac{R_A+\sqrt{R_A(R_A+2p_A)}}{2})$ is the unique value of $\frac{M_B}{c_B}$ at which $u_B^{1A}(p_A,0) = u_B^{1B}(p_A,\hat{p}_B)$.
\end{lemma}

The proof is deferred to Appendix \ref{sec:BR_B}, where we thoroughly analyze the properties of the functions $u_B^{1A}$, $u_B^{2A}$, $u_B^{1B}$, and $u_B^{2B}$ as characterized by Lemma \ref{lem:uB_pB}. This allows us to identify the maximizer $p_B^* \in [0,\frac{M_B}{c_B}]$ of $u_B$ for all possible cases stated in Lemma \ref{lem:uB_pB}.

\subsection{Proof of Proposition \ref{prop:stack_equil}}
We are now ready to establish Proposition \ref{prop:stack_equil}. 

\begin{proof}
Let us define 
\begin{equation}
    \begin{aligned}
        T_1(p_A) &:= p_A + R_A(p_A) \\
        T_2(p_A) &:= p_A + \frac{R_A(p_A) + \sqrt{R_A(p_A) (R_A(p_A) + 2p_A)}}{2}.
    \end{aligned}
\end{equation}
where $R_A(p_A) = M_A-c_Ap_A$. It holds that $T_1(0) = T_2(0) = M_A$, $T_1(\frac{M_A}{c_A}) = T_2(\frac{M_A}{c_A}) = \frac{M_A}{c_A}$, and $T_1(p_A) < T_2(p_A)$ on the interval $p_A \in (0,\frac{M_A}{c_A})$. We prove the result item by item.

\vspace{1mm}

\noindent 1) Suppose $\frac{M_B}{c_B} \leq M_A$. In this case, we have $\frac{M_B}{c_B} \leq T_1(p_A)$ for all $p_A \in [0,\frac{M_A}{c_A}]$, with equality at $p_A = 0$ if and only if $\frac{M_B}{c_B} = M_A$. Thus, player $B$'s best-response against any $p_A$ is $p_B^* = 0$ (Lemma \ref{lem:BR_B}). The scenario reduces to the optimization problem from Corollary 4.1.

\vspace{1mm}

\noindent 2) Suppose $M_A < \frac{M_B}{c_B} \leq \frac{M_A}{c_A}$. 


By Lemma \ref{lem:BR_B}, the threshold at which player $B$'s best-response switches is given by the value of $p_A$ that satisfies $h(M_A-c_Ap_A,p_A) = \frac{M_B}{c_B}$. Equivalently, this is the value of $p_A$ that satisfies $u_B^{1A}(p_A,0) = u_B^{1B}(p_A,\hat{p}_B)$ with $\frac{M_B}{c_B} \leq \frac{1}{c_B}T_2(p_A)$. The latter condition ensures that either $u_B^{1A}(p_A,0)$ or $u_B^{1B}(p_A,\hat{p}_B)$ is player $B$'s best-response payoff (Lemma \ref{lem:uB_pB}). One can write 
\begin{equation}
    u_B^{1B}(p_A,\hat{p}_B) = 1 - \frac{c_A(2-c_B)}{2}\frac{\frac{M_A}{c_A} - p_A}{\frac{M_B}{c_B} - p_A}
\end{equation}
This is a decreasing and concave function on $p_A\in[0,\frac{M_B}{c_B})$, and it decreases to $-\infty$ as $p_A \rightarrow \frac{M_B}{c_B}$. The payoff $u_B^{1A}(p_A,0)$ is given by
\begin{equation}
    u_B^{1A}(p_A,0) = \frac{M_B}{2\sqrt{R_A(p_A)}}\left(\frac{2(T_2(p_A) - p_A)}{T_1(p_A) + 2(T_2(p_A) - p_A)} \right)^2
\end{equation}
This function has a single critical point in the interval $[0,\frac{M_A}{c_A}]$ at $\bar{p}_A = \frac{2(1-c_A)}{2-c_A}\frac{M_A}{c_A}$, which is a local minimum. It is decreasing on $[0,\bar{p}_A)$ and increasing on $(\bar{p}_A,\frac{M_A}{c_A}]$. 



There is a unique value $p_A^\dagger\in (0,\frac{M_B}{c_B})$ such that $\frac{M_B}{c_B} \leq \frac{1}{c_B}T_2(p_A^\dagger)$ where these two functions intersect. Note that from Lemma \ref{lem:BR_B}, $p_A^\dagger$ is the unique value that satisfies 
\begin{equation}
    h(M_A-c_Ap_A^\dagger,p_A^\dagger) = \frac{M_B}{c_B}.
\end{equation}
Player $B$'s best-response is $p_B^* = \hat{p}_B$ for $p_A < p_A^\dagger$, and $p_B^* = 0$ for $p_A > p_A^\dagger$  (Lemma \ref{lem:BR_B}). Consequently, player $A$'s payoff (under player $B$'s best-response) is given by
\begin{equation}
    u_A(p_A,p_B^*) = 
    \begin{cases}
        1 - u_B^{1B}(p_A,\hat{p}_B), &\text{if } p_A \in [0,p_A^\dagger) \\
        1 - u_B^{1A}(p_A,0), &\text{if } p_A \in (p_A^\dagger,\frac{M_A}{c_A}]
    \end{cases}
\end{equation}
We observe that $u_A(p_A,p_B^*)$ is increasing on $p_A \in [0,p_A^\dagger)$. If $p_A^\dagger \geq \bar{p}_A$, then $u_A(p_A,p_B^*)$ is decreasing on $(p_A^\dagger,\frac{M_A}{c_A}]$, and hence player $A$'s security strategy is $p_A^* = p_A^\dagger$. The resulting payoff to player $A$ is 
\begin{equation}
    1 - u_B^{1B}(p_A^\dagger,\hat{p}_B) = \frac{c_A(2-c_B)}{2}\frac{\frac{M_A}{c_A}-p_A^\dagger}{\frac{M_B}{c_B}-p_A^\dagger}.
\end{equation}

If $p_A^\dagger < \bar{p}_A$, then $u_A(p_A,p_B^*)$ is increasing on $(p_A^\dagger,\bar{p}_A]$ and decreasing on $(\bar{p}_A,\frac{M_A}{c_A}]$. Hence, player $A$'s security strategy is $p_A^* = \bar{p}_A = \frac{2(1-c_A)}{2-c_A}\frac{M_A}{c_A}$. The resulting payoff to player $A$ is given by Corollary 4.1.

\vspace{1mm}

\noindent 3) Suppose $\frac{M_B}{c_B} > \frac{M_A}{c_A}$.

In this case, the function $u_B^{1B}(p_A,\hat{p}_B)$ is strictly increasing on $p_A \in [0,\frac{M_A}{c_A}]$. Any intersection (if any) with $u_B^{1A}(p_A,0)$ must occur on $p_A > \bar{p}_A$, where $u_B^{1A}(p_A,0)$ is increasing. Therefore, $u_A(p_A,p_B^*)$ must be strictly decreasing on $[0,\frac{M_A}{c_a}]$ and hence $p_A^* = 0$. The resulting payoff to player $A$ is 
\begin{equation}
    1 - u_B^{1B}(0,\hat{p}_B) =  \frac{c_B(2-c_B)}{2}\frac{M_A}{M_B}.
\end{equation} 
\end{proof}

\section{Conclusion}

In this manuscript, we studied the strategic role of pre-allocations in competitive interactions under a two-stage General Lotto game model, where one of the players can place resources before a decisive point of conflict. Our main contribution fully provided subgame-perfect equilibrium characterizations to this formulation. This result revealed a rich interplay between the effectiveness of pre-allocated resources and real-time resources, which allowed us to quantify optimal investments in both types of resources. We then analyzed a Stackelberg game scenario where the other player is able to respond to the pre-allocation before the final decisive round. This highlights the significance that more dynamic and sequential interactions can have on a player's eventual performance. Future work will involve studying these dynamic interactions in richer environmental contexts, e.g. with multiple fronts of battlefields.

\bibliographystyle{abbrv}
\bibliography{source}

\appendix

\subsection{Proof of Part 2-b}

Here, we present the proof of Part 2-b from the proof outline of Theorem \ref{thm:equilibrium_characterization}. It states that: \emph{Any pre-allocation $\bs{p}$ that corresponds to a solution of \eqref{eq:SOE} with $\mcal{B}_1,\mcal{B}_2 \neq \varnothing$ satisfies $\pi_A(\bs{p},R_A,R_B) \leq \pi_A(\bs{p}^*,R_A,R_B)$.}


Throughout the proof, we will use the short-hand notation $W_j=\sum_{b\in\mcal{B}_j} w_b$, $P_j=\sum_{b\in\mcal{B}_j} p_b$ and $\bs{p}_j=(p_b)_{b\in\mcal{B}_j}$, for $j=1,2$. For $\bs{p}\in\Delta_n(P)$, we obtain the system of equations
\begin{equation*}
\begin{aligned}
    R_A &= \sum_{b\in\mcal{B}_1} \frac{(w_b \kappa_B-p_b)^2}{2 w_b \kappa_B} 
            + \sum_{b\in\mcal{B}_2} \frac{(w_b \kappa_A)^2}{2 w_b \kappa_B}, \\
    R_B &= \sum_{b\in\mcal{B}_1} \frac{( w_b \kappa_B)^2-(p_b)^2}{2  w_b \kappa_A} 
            + \sum_{b\in\mcal{B}_2} \frac{(w_b \kappa_A+p_b)^2-(p_b)^2}{2  w_b \kappa_A},
\end{aligned}
\end{equation*}
where $0<w_b\kappa_B-p_b\leq w_b\kappa_A$ holds for all $b\in\mcal{B}_1$, and $ w_b \kappa_B-p_b>w_b\kappa_A$ holds for all $b\in\mcal{B}_2$.  The system of equations readily gives the expression:
\begin{equation} \label{eq:mixedcases_systemeq}
\begin{aligned}
    W_1 \kappa_B^2 + W_2 \kappa_A^2 &= 2\kappa_B (X_A+P_2) - ||\bs{p}_1||^2_w \\
    &= 2\kappa_A (X_B-P_2) + ||\bs{p}_1||^2_w,
\end{aligned}
\end{equation}
where recall that $||\bs{p}_1||^2_w = \sum_{b\in\mcal{B}_2} [(p_b)^2/w_b]$.  The solution to the above system of equations is
\begin{equation} \label{eq:mixedcases_solution}
\begin{aligned}
    \kappa^*_B &= \frac{C_1 H_2 \pm \sqrt{C_2^2H_1H_2}}{W_1 C_2^2 + W_2 C_1^2}, \\
    \kappa^*_A  &= \frac{C_2 H_1 \pm \sqrt{C_1^2H_1H_2}}{W_1 C_2^2 + W_2 C_1^2},
\end{aligned}
\end{equation}
where we denote $C_1:=R_A+P_1$, $C_2:=R_B-P_2$, $H_1:=C_1^2-W_1||\bs{p}_1||^2_{\bs{w}}$ and $H_2:=C_2^2+W_2||\bs{p}_1||^2_{\bs{w}}$.  We consider only the scenario where $\pm=+$ in \eqref{eq:mixedcases_solution}, since the expression for $\kappa^*_A$ is strictly negative when $\pm=-$.  Simply observe that $C_1>0$, $C_1^2 > H_1$, $0<C_2^2 < H_2$ and, thus, that either (i) $H_1>0$, $C_2>0$ and $0<C_2 H_1 < C_1\sqrt{H_1 H_2}$, (ii) $H_1<0$, $C_2<0$ and $0< C_2 H_1 = |C_2| |H_1| < C_1 \sqrt{|H_1| |H_2|}$, or (iii) only one of $H_1$ or $C_2$ is negative, in which case $C_2 H_1 < 0$.

Substituting \eqref{eq:mixedcases_solution} into \eqref{eq:playerA_payoff} and simplifying, we obtain
\begin{equation}
\begin{aligned}
    \pi_A(\bs{p},R_A,R_B) 
    &= W_1 + \frac{\sqrt{H_1H_2}-C_1C_2}{||\bs{p}_1||^2_{\bs{w}}},
\end{aligned}
\end{equation}
and the partial derivatives of $\pi_A(\bs{p},R_A,R_B)$ with respect to $p_b$ are as follows. For $b\in\mcal{B}_1$,
\begin{equation}\label{eq:mixed_partials1}
    \begin{aligned}
        \frac{\partial \pi_A}{\partial p_b} &= 
        \frac{-p_b/w_b}{(||\bs{p}_1||^2_{\bs{w}})^2\sqrt{H_1H_2}} (C_1\sqrt{H_2}-C_2\sqrt{H_1})^2 \\
        &\quad +\frac{1}{||\bs{p}_1||^2_{\bs{w}}\sqrt{H_1}}(C_1\sqrt{H_2}-C_2\sqrt{H_1})
    \end{aligned}
\end{equation}
and for $b\in\mcal{B}_2$,
\begin{equation}\label{eq:mixed_partials2}
    \frac{\partial  \pi_A}{\partial p_b} = 
        \frac{1}{||\bs{p}_1||^2_{\bs{w}}\sqrt{H_2}}(C_1\sqrt{H_2}-C_2\sqrt{H_1}).
\end{equation}
We first consider critical points $\bs{p}$ strictly in the interior of $\Delta_n(P)$, and resolve the points on the boundary later.  One necessary condition for a critical point is that $(\frac{\partial}{\partial p_b} - \frac{\partial}{\partial p_c})\pi_A = 0$ for all $b\in\mcal{B}_1$ and $c\in\mcal{B}_2$.  Firstly, observe that $C_1 > \sqrt{H_1}$ and $\sqrt{H_2}>C_2$, and, thus, it must be that $C_1\sqrt{H_2}-C_2\sqrt{H_1}>0$.  We can thus divide the expression $\frac{\partial \pi_A}{\partial p_b} = \frac{\partial \pi_A}{\partial p_c}$ on both sides by $C_1\sqrt{H_2}-C_2\sqrt{H_1}$ and rearrange to obtain
\[ (p_b/w_b) (C_1\sqrt{H_2}-C_2\sqrt{H_1})
        = ||\bs{p}_1||^2_{\bs{w}} (\sqrt{H_2} - \sqrt{H_1}) > 0. \]
Observe that the left-hand side is strictly greater than zero, and, thus, the right-hand side must be as well.  This immediately requires $\sqrt{H_2}-\sqrt{H_1}>0$, since $||\bs{p}_1||^2_{\bs{w}}>0$.  Re-arranging the above expression, note that we also require
\[ \sqrt{H_1} [C_2 (p_b/w_b) - ||\bs{p}_1||^2_{\bs{w}}] = \sqrt{H_2} [ C_1 (p_b/w_b) - ||\bs{p}_1||^2_{\bs{w}}]. \]
Since we have just shown that $\sqrt{H_2}>\sqrt{H_1}$ must hold, it follows that each $b\in\mcal{B}_1$ satisfies either (i) $C_2 (p_b/w_b) - ||\bs{p}_1||^2_{\bs{w}}<C_1 (p_b/w_b) - ||\bs{p}_1||^2_{\bs{w}}<0$; or (ii) $C_2 (p_b/w_b) - ||\bs{p}_1||^2_{\bs{w}} > C_1 (p_b/w_b) - ||\bs{p}_1||^2_{\bs{w}} > 0$.  
Observe that $C_1 (p_b/w_b) > ||\bs{p}_1||^2_1$ must hold for $b' \in \arg\,\max_{b\in\mcal{B}_1} p_b/w_b$, and thus $b'$ must satisfy scenario (ii) and $C_2>C_1$ (or, equivalently, $R_B-P>R_A$). This last inequality then implies that scenario (ii) must be satisfied for all $b\in\mcal{B}_1$.

We have shown that, in order for $(\frac{\partial}{\partial p_b} - \frac{\partial}{\partial p_c}) \pi_A=0$ to hold for all $b\in\mcal{B}_1$ and $c\in\mcal{B}_2$, a critical point $\bs{p}$ must satisfy
\[ \frac{p_b}{w_b} = \bar p := \frac{\sqrt{H_2} - \sqrt{H_1}}{C_1\sqrt{H_2}-C_2\sqrt{H_1}}||\bs{p}_1||^2_w, \]

\noindent for each $b\in\mcal{B}_1$.  Expanding this expression, and solving for $\bar p$ explicitly, we obtain the following two possible (real) solutions for $\bar p$:
\[ \bar p = 0 \text{ or } \bar p = \frac{2( R_B - P)(q R_B-R_A-P)}{R_A}, \]
where we use $P_1=W_1 \bar p$, $P_2=P-P_1$, and $||\bs{p}_1||^2_{\bs{w}}=W_1 \bar{p}^2$.  As $\bar p=0$ is inadmissible, we consider the latter expression for $\bar p$.  After inserting this expression for $\bar p$ into the right-hand side of \eqref{eq:case3_mostrestrictive}, where $\max_b\{p_b/w_b\}=\bar p$, we obtain
\begin{equation*}
    \frac{R_A}{2}\left[1 + \sqrt{1+\frac{2W}{R_A}\bar p} \right] = \frac{R_A}{2} + R_B-P - \frac{R_A}{2} = R_B -P,
\end{equation*}
which follows since we showed above that $R_B-P>R_A$ must hold.  Thus, the only critical point sits at the boundary of the region where all battlefields are in Case 2, since decreasing $\bar p$ even slightly will satisfy the condition in \eqref{eq:case3_mostrestrictive}.  We can further verify that the payoff at this critical point is equal to the constant payoff in the region where all battlefields are in Case 2, but omit this for conciseness.

We conclude the proof by resolving the scenario where $\bs{p}$ lies on the boundaries of $\Delta_n(P)$.  Observe that the conditions on $q\kappa^*_B$ and $\kappa^*_A$ immediately imply that $p_b/w_b > p_c/w_c$ for any $b\in\mcal{B}_1$ and $c\in\mcal{B}_2$.  Thus, on the boundaries of $\Delta_n(P)$, it must either be that all battlefields with $p_b=0$ (and possibly more) are in Case 2, or that all battlefields in $\mcal{B}$ are in Case 1 (which is covered by Lemma \ref{lem:global_maximizer}, part 2-a).  

In the scenario where all battlefields with $p_b=0$ are in Case 2, note that the necessary condition $(\frac{\partial}{\partial p_i} -\frac{\partial}{\partial p_j}) \pi_A \geq 0$ for $i\in\arg\,\min_{b\in\mcal{B}_1} \{p_b/w_b\}$ and $j\in\arg\,\max_{b\in\mcal{B}_1} \{p_b/w_b\}$ only holds with equality if $p_b/w_b = P_1/W_1$ for all $b\in\mcal{B}_1$.  If $P_1/W_1 < \bar p$, then the inequality in \eqref{eq:case3_mostrestrictive} is satisfied implying that all battlefields are in Case 2, and Lemma \ref{lem:global_maximizer} shows that $\bs{p}^*$ must correspond with the same payoff to player $A$.  Otherwise, if $P_1/W_1 = \bar p$, then we showed above that the global maximum sits at the boundary where all battlefields are in Case 2 and $\bs{p}^*$ achieves the same payoff.  

Finally, if $P_1/W_1 > \bar p$, then, from \eqref{eq:mixed_partials1} and \eqref{eq:mixed_partials2}, we know that $(\frac{\partial}{\partial p_b} -\frac{\partial}{\partial p_c}) \pi_A < 0$ must hold for all $b\in\mcal{B}_1$ and $c\in\mcal{B}_2$, since the choice $p_b/w_b = \bar p$ satisfies $(\frac{\partial}{\partial p_b} -\frac{\partial}{\partial p_c}) \pi_A = 0$, and $\frac{\partial \pi_A}{\partial p_b}$ is decreasing with respect to $p_b/w_b$ while $\frac{\partial \pi_A}{\partial p_c}$ is constant.  This violates a necessary condition for a critical point, and implies that $A$'s payoff is increasing in the direction of decreasing $p_b$ and increasing $p_c$, as expected.  




\subsection{Proof of Lemma \ref{lem:uB_pB}}\label{sec:follower_lemma}

We prove the result in two parts, considering the two separate intervals $p_B \in[0,p_A]$ and $p_B \in(p_A,\frac{M_B}{c_B}]$.

\noindent\textbf{\underline{Part 1}: } $p_B \in[0,p_A]$. Define the function
\begin{equation}
    G(p_B) = \frac{2(M_B - p_A + p_B(1-c_B))^2}{2(M_B - p_A + p_B(1-c_B)) + p_A - p_B}
\end{equation}
The definition of $\mc{R}^{1A}$ may be restated as
\begin{equation}
    \mc{R}^{1A} = \left\{ p_B < p_0, \text{ or } p_B \geq p_0 \text{ and } G(p_B) < R_A \right\}.
\end{equation}
where $p_0 = \frac{p_A-M_B}{1-c_B}$.

\noindent $\bullet$ Suppose $\frac{M_B}{c_B} \leq R_A+p_A$. Since $M_B \leq c_Bp_A$, we have $p_0 > p_A$, and consequently $u_B(\tbfp) = u_B^{1A}(\tbfp)$ $\forall p_B \in [0,p_A]$. The sign of $\frac{\partial u_B^{1A}}{\partial p_B}$ is equivalent to the sign of 
\begin{equation}
    \begin{aligned}
        -f_A(p_B)&\left[ c_B(2p_A^2+R_A(2R_A-3p_B)) \right. \\ 
        &+ p_A(c_B(5R_A-2p_B)-2M_B) \\
        &+(f_A(p_B)-R_A)(c_B(3p_A+2R_A-p_B)-2M_B) \\ 
        &\left.- 2M_B(R_A-p_B) \right]
    \end{aligned}
\end{equation}
where $f_A(p_B) = R_A + \sqrt{R_A(R_A+2(p_A-p_B))}$. It then holds that $\frac{\partial u_B^{1A}}{\partial p_B}(p_A,p_B) \leq 0$ for $p_B \in [0,p_A]$ if and only if $\frac{M_B}{c_B} \leq R_A+p_A$. It follows that $u_B$ is decreasing on the interval $p_B \in [0,p_A]$.

\vspace{1mm}

\noindent $\bullet$ Suppose $R_A+p_A < \frac{M_B}{c_B} \leq \frac{R_A}{c_B} + p_A$. First we consider the case that $\frac{R_A}{c_B} + p_A < \frac{p_A}{c_B}$. Here, $p_0 \in [0,p_A)$ and thus we immediately have $u_B(\tbfp) = u_B^{1A}(\tbfp)$ for $p_B \in [0,p_0)$. We also have $G(p_0) = 0$, $G(p_B)$ is strictly increasing on $p_B \in (p_0,p_A)$, and $G(p_A) < R_A$ (from $M_B \leq R_A + c_BpA$). Therefore, $G(p_B) < R_A$ for all $p_B \in (p_0,p_A]$, and hence $u_B(\tbfp) = u_B^{1A}(\tbfp)$.

Now, we consider the case that $\frac{R_A}{c_B} + p_A \geq \frac{p_A}{c_B}$. The range $p_A < \frac{M_B}{c_B} \leq \frac{p_A}{c_B}$ follows the above argument identically. Now, when $\frac{p_A}{c_B} < \frac{M_B}{c_B} \leq \frac{R_A}{c_B} + p_A$, we have $p_0 < 0$, $G(p_B)$ is strictly increasing on $[0,p_A]$, and $G(p_A) < R_A$. Therefore, $G(p_B) < R_A$ for all $p_B \in (p_0,p_A]$, and hence $u_B(\tbfp) = u_B^{1A}(\tbfp)$.

\noindent $\bullet$ Suppose $\frac{R_A}{c_B} + p_A < \frac{M_B}{c_B} \leq \frac{1}{c_B}\left(p_A + \frac{R_A + \sqrt{R_A(R_A + 2p_A)}}{2}\right)$. First we consider the case $\frac{R_A}{c_B} + p_A < \frac{p_A}{c_B}$. In the range $\frac{R_A}{c_B} + p_A < \frac{M_B}{c_B} \leq \frac{p_A}{c_B}$, we have $p_0 \in [0,p_A)$ and thus we immediately have $u_B(\tbfp) = u_B^{1A}(\tbfp)$ for $p_B \in [0,p_0)$. We also have $G(p_0) = 0$ and $G(p_B)$ is strictly increasing on $p_B \in (p_0,p_A)$. Now, the condition $\frac{R_A}{c_B} + p_A < \frac{M_B}{c_B}$ asserts that $G(p_A) > R_A$. Thus, there is a unique value, $p_B^{1A} \in (p_0,p_A)$ that solves $G(p_B^{1A}) = R_A$. 

In the range $\frac{p_A}{c_B} < \frac{M_B}{c_B} \leq \frac{1}{c_B}\left(p_A + \frac{R_A + \sqrt{R_A(R_A + 2p_A)}}{2}\right)$, $G(p_B)$ is strictly increasing, $G(0) < R_A$, and $G(p_A) > R_A$ which yields the result. To see this, we first note that $p_0 < 0$. To show $G$ is strictly increasing, the sign of $G'(p_B)$ is equal to the sign of $ (2c_B^2-3c_B+1)p_B + (3-2c_B)M_B - p_A$ (from $M_B > p_A$). Observe that this value is linear in $p_B$. At $p_B = 0$, it is $(3-2c_B)M_B - p_A > 0$, and at $p_B = p_A$, it is $(3-2c_B)(M_B-c_Bp_A) > 0$. Thus, $G'(p_B) > 0$ for $p_B \in (0,p_A)$. 
The condition $G(0) = \frac{2(M_B-p_A)^2}{2(M_B-p_A)+p_A} < R_A$ is equivalent to $M_B \leq \frac{1}{c_B}\left(p_A + \frac{R_A + \sqrt{R_A(R_A + 2p_A)}}{2}\right)$.

The case $R_A + c_Bp_A \geq p_A$ follows identical arguments to the above.

\noindent $\bullet$ Suppose $\frac{M_B}{c_B} > \frac{1}{c_B}\left(p_A + \frac{R_A + \sqrt{R_A(R_A + 2p_A)}}{2}\right)$. Here, $p_0 < 0$, $G(p_B)$ is strictly increasing on $(0,p_A)$, and $G(p_A) < R_A$. This yields the result.

\noindent\textbf{\underline{Part 2}: } $p_B \in(p_A,\frac{M_B}{c_B}]$. Define the functions 
    \begin{equation}
        \begin{aligned}
            F(p_B) &:= \frac{2(t_A-p_B)^2}{2R_A+p_A-p_B} \\
            L(p_B) &:= M_B - c_Bp_B
        \end{aligned}
    \end{equation}
    where $t_A := R_A + p_A$. The definition of $\mc{R}^{1B}$ may be restated as
    \begin{equation}
        \mc{R}^{1B} = \{ p_B > t_A, \text{ or } p_B \leq t_A \text{ and } F(p_B) \leq L(p_B) \}.
    \end{equation}
    
    \noindent $\bullet$ Suppose $\frac{M_B}{c_B} \leq t_A$. Here, $F(p_B)$ is strictly decreasing on $(p_A,\frac{M_B}{c_B})$, $F(p_A) = R_A > L(p_A)$, and $F(\frac{M_B}{c_B}) > L(\frac{M_B}{c_B} = 0$. The solutions to the equation $F(p_B) = L(p_B)$ are given by
    \begin{equation}
        \footnotesize 
        \begin{aligned}
            &r_\pm(M_B) := \frac{1}{2(2-c_B)}\left[2t_A(2-c_B) \!-\! L(p_A) \right. \\ 
            &\left. \pm \sqrt{(2t_A(2-c_B) - L(p_A))^2 \!-\! 4(2-c_B)(2t_A(t_A\!-\!M_B) + M_Bp_A)} \right]
        \end{aligned}
    \end{equation}
    The roots are complex if and only if $M_B < c_Bp_A - 2(2-c_B)R_A + 2R_A\sqrt{2(2-c_B)}$, in which case $u_B = u_B^{2B}$ for all $p_B \in [p_A,\frac{M_B}{c_B}]$ (since $F(p_B) > L(p_B)$).
    
    So, we consider $c_Bp_A - 2(2-c_B)R_A + 2R_A\sqrt{2(2-c_B)} \leq M_B \leq c_Bt_A$, in which the roots are real-valued, and only $r_+$ can be in the interval $(p_A,\frac{M_B}{c_B})$. We thus have
    \begin{equation}
        u_B(\tbfp) =
        \begin{cases}
            u_B^{2B}(\tbfp), &p_B \in [p_A,r_-] \\
            u_B^{1B}(\tbfp), &p_B \in (r_-,r_+] \\
            u_B^{2B}(\tbfp), &p_B \in (r_+,\frac{M_B}{c_B}]
        \end{cases}
    \end{equation}
    
    Both payoff functions $u_B^{2B}$ and $u_B^{1B}$ are decreasing in $p_B$ in their respective intervals. Indeed,
    \begin{equation}
        \frac{\partial u_B^{2B}}{\partial p_B} = \frac{L(p_B) - c_B(t_A - p_B)}{2(t_A - p_B)^2} \leq 0
    \end{equation} 
    where the inequality is due to $M_B \leq c_Bt_A$. 
    
    To show $u_B^{1B}$ is decreasing, we observe that the payoff $u_B^{1B}$ admits a single critical point in the interval $(p_A,\frac{M_B}{c_B})$. Indeed, denoting $\ell(p_B) = \sqrt{L(p_B)} \geq 0$ and $g(p_B) = \sqrt{L(p_B) + 2(p_B-p_A)} > 0$, we can calculate
    \begin{equation}
        \frac{\partial u_B^{1B}}{\partial p_B} = -\frac{R_A}{2}\frac{(\ell+g)(c_B(\ell+g) - 2\ell)}{\ell g(\ell g+L(p_B) + p_B - p_A)^2}.
    \end{equation}
    It holds that $\frac{\partial u_B^{1B}}{\partial p_B}(p_A) > 0$, $\frac{\partial u_B^{1B}}{\partial p_B}(\frac{M_B}{c_B}) < 0$, and $\frac{\partial u_B^{1B}}{\partial p_B}(p_B) = 0$ if and only if 
    \begin{equation}
        p_B = \hat{p}_B := \frac{M_B}{c_B} - \frac{M_B-c_Bp_A}{2-c_B} \in (p_A,\frac{M_B}{c_B}).
    \end{equation}
    Thus, $\hat{p}_B$ is a local maximizer, with $u_B^{1B}$ decreasing for $p_B > \hat{p}_B$. Now, observe that $r_-(c_Bt_A) = \hat{p}_B$, $r_+(c_Bt_A) = \frac{M_B}{c_B}$, $r_-(M_B)$ is decreasing in $M_B$, and $\hat{p}_B$ is increasing in $M_B$. We have $\hat{p}_B \leq r_-$, and therefore $u_B(\tbfp)$ is decreasing for all $p_B \in [p_A,\frac{M_B}{c_B}]$.
        
    \noindent $\bullet$ Suppose $t_A < \frac{M_B}{c_B} \leq \frac{R_A}{c_B}+p_A$. Here,  $u_B = u_B^{1B}$ for $p_B \in (t_A,\frac{M_B}{c_B})$. The function $F(p_B)$ is strictly decreasing on $(p_A,t_A)$ with $F(p_A) = R_A \geq L(p_A)$ and $F(t_A) = 0$. Thus, we have $F(p_B) > L(p_B)$ for $p_B \in [p_A,p_B^{1B}]$ and $F(p_B) \leq L(p_B)$ for $p_B \in [p_B^{1B},t_A]$, where $p_B^{1B}$ is the unique value in the interval $[p_A,t_A)$ that satisfies $F(p_B) = L(p_B)$. We thus get the stated characterization for $u_B(\tbfp)$.
        
    \noindent $\bullet$ Suppose $\frac{R_A}{c_B}+p_A < M_B$. Here, $u_B = u_B^{1B}$ for $p_B \in (t_A,\frac{M_B}{c_B})$. The function $F(p_B)$ is strictly decreasing on $(p_A,t_A)$ with $F(p_A) = R_A < M_B - c_Bp_A$ and $F(t_A) = 0$. Thus, we have $F(p_B) < M_B - c_Bp_B$ for $p_B \in [p_A,t_A]$. Therefore, $u_B = u_B^{1B}$ for all $p_B \in [p_A,\frac{M_B}{c_B}]$.
    
    From identical arguments from the first bullet point ($\frac{M_B}{c_B} \leq t_A$), we know that $u_B^{1B}$ admits the local maximizer $\hat{p}_B \in (p_A,\frac{M_B}{c_B})$. Combining the characterizations from Part 1 and Part 2 yields the result. 

\subsection{Proof of Lemma \ref{lem:BR_B}}\label{sec:BR_B}

We first address the two extreme lower and upper intervals. We will denote $f_A(p_B) = R_A + \sqrt{R_A(R_A+2(p_A-p_B))}$.
    \begin{itemize}[leftmargin=*]
        \item Suppose $\frac{M_B}{c_B} \leq R_A+p_A$. By Lemma \ref{lem:BR_B} $u_B(p_A,p_B)$ is  decreasing for all $p_B \in [0,\frac{M_B}{c_B}]$, and hence $\max_{p_B \in [0,\frac{M_B}{c_B}]} u_B(p_A,p_B) = u_B(0)$.
        
        \item Suppose $p_A+\frac{f_A(0)}{2} < \frac{M_B}{c_B}$. Here, $u_B$ is strictly increasing on $p_B \in [0,p_A)$, and the maximum value in the interval $p_B \in (p_A,\frac{M_B}{c_B}]$ is given by $u_B^{1B}(\hat{p}_B)$. Therefore,  
        \begin{equation}
            \max_{p_B \in [0,\frac{M_B}{c_B}]} u_B(p_B) = u_B^{1B}(\hat{p}_B).
        \end{equation}

        \item Now, consider the middle interval $R_A+p_A < \frac{M_B}{c_B} \leq p_A+\frac{f_A(0)}{2}$. By Lemma \ref{lem:uB_pB}, $u_B(0) = u_B^{1A}(0)$. Here, $u_B$ is decreasing at $p_B = 0$ but is increasing at $p_B = p_A$ (whether in region $\mc{R}^{1A}$ or $\mc{R}^{2A}$). By Lemma \ref{lem:uB_pB}, the maximum value in the interval $p_B \in (p_A,\frac{M_B}{c_B}]$ is given by $u_B^{1B}(\hat{p}_B)$.  Therefore, 
        \begin{equation}
            \max_{p_B \in [0,\frac{M_B}{c_B}]} u_B(p_B) = \max\{u_B^{1A}(0),u_B^{1B}(\hat{p}_B) \}.
        \end{equation}
        
        We can identify the existence of a threshold value $h(R_A,p_A) \in (R_A+p_A,p_A+\frac{f_A(0)}{2})$ for which $p_B^* = 0$ for $\frac{M_B}{c_B}< h(p_A)$, $p_B^* = \hat{p}_B$ for $\frac{M_B}{c_B}> h(p_A)$, and $p_B^* = \hat{p}_B$ or $0$ for $\frac{M_B}{c_B}= h(p_A)$. Indeed the payoff $u_B^{1B}(\hat{p}_B)$ may be written as
        \begin{equation}
            u_B^{1B}(\hat{p}_B) = 1 - \frac{c_B(2-c_B)}{2}\frac{R_A}{M_B - c_Bp_A}
        \end{equation}
        The payoff $u_B^{1A}(0)$ is given by
        \begin{equation}
            u_B^{1A}(0) = \frac{M_B}{2}\left(\frac{\sqrt{R_A} + \sqrt{R_A+2p_A}}{R_A+p_A+\sqrt{R_A(R_A+2p_A)}} \right)^2
        \end{equation}
        Denoting $K_1 = \frac{c_B(2-c_B)}{2}$ and $K_2 = \left(\frac{\sqrt{R_A} + \sqrt{R_A+2p_A}}{R_A+p_A+\sqrt{R_A(R_A+2p_A)}} \right)^2$, we have $u_B^{1B}(\hat{p}_B) = u_B^{1A}(0)$ precisely at
        \begin{equation}
            \begin{aligned}
                &\frac{M_B}{c_B} = \frac{1}{c_B K_2}\left[\frac{K_2}{2}c_Bp_A + 1 \right. \\ 
                &\quad\quad\left. - \sqrt{(\frac{K_2}{2}c_Bp_A + 1)^2 - 2K_2(c_Bp_A+K_1R_A)} \right]
            \end{aligned}
        \end{equation}
        Denoting $h(R_A,p_A)$ as the right-hand side above, we obtain the result.
    \end{itemize}

\begin{IEEEbiography}[{\includegraphics[width=1in,height=1.25in,clip,keepaspectratio]{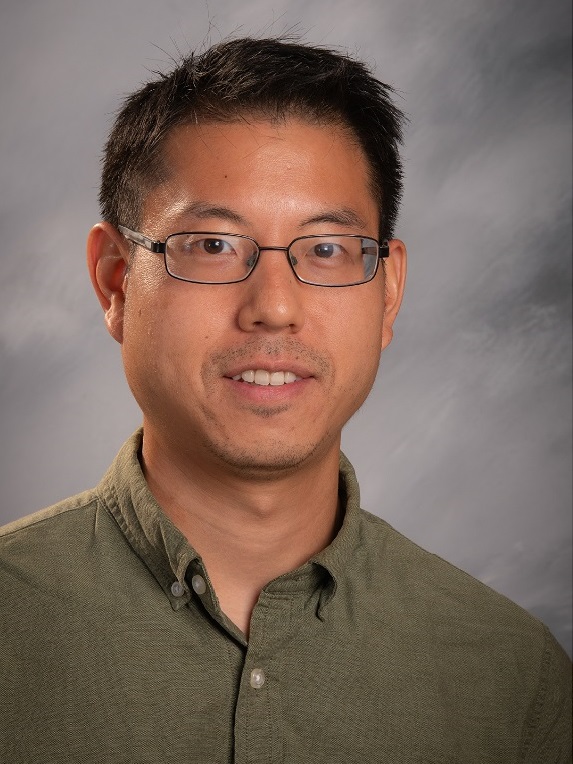}}] {Keith Paarporn}
  is an Assistant Professor in the Department of Computer Science at the University of Colorado, Colorado Springs. He received a B.S. in Electrical Engineering from the University of Maryland, College Park in 2013, an M.S. in Electrical and Computer Engineering from the Georgia Institute of Technology in 2016, and a Ph.D. in Electrical and Computer Engineering from the Georgia Institute of Technology in 2018. From 2018 to 2022, he was a postdoctoral scholar in the Electrical and Computer Engineering Department at the University of California, Santa Barbara. His research interests include game theory, control theory, and their applications to multi-agent systems and security. 
\end{IEEEbiography}

\begin{IEEEbiography}[{\includegraphics[width=1in,height=1.25in,clip,keepaspectratio]{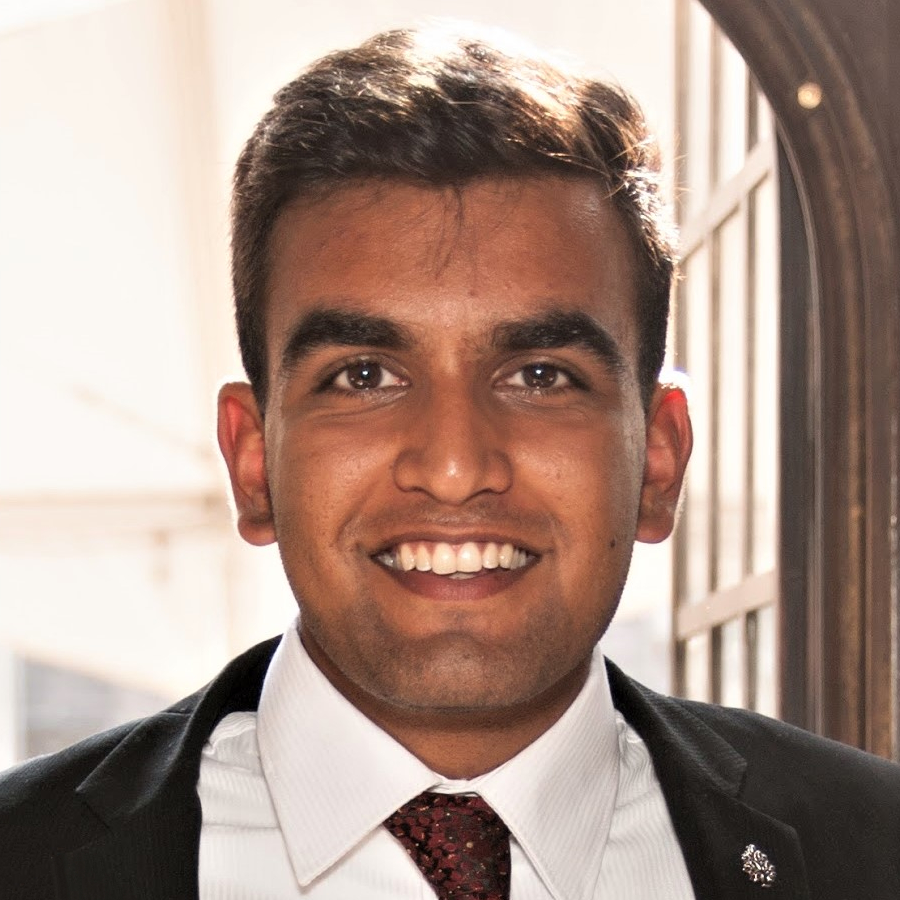}}] {Rahul Chandan}
  is a Research Scientist at Amazon Robotics. He received a B.A.Sc. in  Engineering Science from the University of Toronto in 2017, and an M.S. and PhD in Electrical and Computer Engineering from the University of California, Santa Barbara in 2019 and 2022, respectively. His research interests include game theory, multi-agent systems, optimization, and economics and computation. 
\end{IEEEbiography}

\begin{IEEEbiography}[{\includegraphics[width=1in,height=1.25in,clip,keepaspectratio]{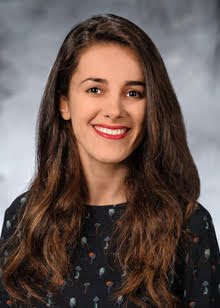}}] {Mahnoosh Alizadeh}
  is an Associate Professor of Electrical and Computer Engineering at the University of California Santa Barbara. She received the B.Sc. degree ('09) in Electrical Engineering  from Sharif University of Technology and the M.Sc. ('13) and Ph.D. ('14) degrees in Electrical and Computer Engineering from the University of California Davis, where she was the recipient of the Richard C. Dorf award for outstanding research accomplishment. From 2014 to 2016, she was a postdoctoral scholar at Stanford University. Her research interests are focused on designing scalable control, learning, and market mechanisms to promote efficiency and resiliency in societal-scale cyber-physical systems. Dr. Alizadeh is a recipient of the NSF CAREER award and the best paper award from HICSS-53 power systems track. 
\end{IEEEbiography}

\begin{IEEEbiography}[{\includegraphics[width=1in,height=1.25in,clip,keepaspectratio]{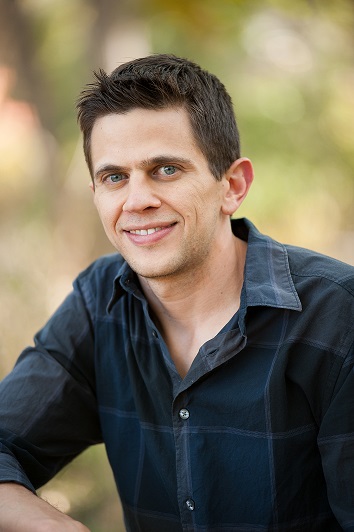}}] {Jason Marden}
    is a Professor in the Department of Electrical and Computer Engineering at
    the University of California, Santa Barbara. Jason received a BS in Mechanical Engineering in 2001 from
    UCLA, and a PhD in Mechanical Engineering in
    2007, also from UCLA, under the supervision of Jeff
    S. Shamma, where he was awarded the Outstanding
    Graduating PhD Student in Mechanical Engineering.
    After graduating from UCLA, he served as a junior
    fellow in the Social and Information Sciences Laboratory at the California Institute of Technology until
    2010 when he joined the University of Colorado. Jason is a recipient of the NSF Career Award (2014), the ONR Young Investigator Award (2015), the AFOSR Young Investigator Award (2012), the American Automatic Control Council Donald P. Eckman Award (2012), and the SIAG/CST Best SICON Paper Prize (2015). Jason’s research interests focus on game theoretic methods for the control of distributed multiagent systems.
\end{IEEEbiography}

\end{document}